\documentclass[11pt,reqno]{amsart}

\usepackage[foot]{amsaddr}

\usepackage{hyperref}
\usepackage{amsmath}
\usepackage{amsthm}
\usepackage{amsfonts}
\usepackage{amssymb}
\usepackage{cleveref}

\usepackage{latexsym,amsmath}
\usepackage{geometry}
\usepackage{fixmath}

\usepackage[dvips]{graphicx}
\usepackage{graphicx}
\usepackage{fullpage}
\usepackage{appendix}

\usepackage{color} 
\usepackage{times}

\usepackage{leftidx}
\usepackage{scrextend}
\usepackage{float}
\usepackage{graphicx}
\usepackage[symbol]{footmisc}
\usepackage[shortlabels]{enumitem}

\usepackage{bm}

\newtheorem{theorem}{Theorem}[section]
\newtheorem{lemma}[theorem]{Lemma}

\newtheorem{proposition}[theorem]{Proposition}

\newtheorem{corollary}[theorem]{Corollary}

\newtheorem{definition}[theorem]{Definition}

\usepackage[textsize=tiny]{todonotes}

\newcommand{\cE}{{\mathcal{E}}}

\newcommand{\cI}{{\mathcal{I}}}

\newcommand{\bG}{\mathbold{G}}

\newcommand{\bS}{\mathbold{S}}

\newcommand{\tuple}[1]{\left(#1\right)} 
\newcommand{\tp}{\tuple}
\newcommand{\entropy}{{\rm Ent}}
\newcommand{\var}{{\rm Var}}
\newcommand{\saw}{\mathrm{SAW}}

\newcommand{\blambda}{\mathbold{\lambda}}

\newcommand{\bsigma}{{\mathbold{\sigma}}}

\newcommand{\lnorm}{\left | \left |}
\newcommand{\rnorm}{\right| \right|}

\newcommand{\Glauber}{{\{X_t\}_{t\geq 0}}}
\date{}

\begin{document}

\title{On the Mixing Time of Glauber Dynamics \\ for the Hard-core and Related Models   on   $G(n,d/n)$}
\author{Charilaos Efthymiou$^*$ and Weiming Feng$^\dagger$}
\thanks{$^*$ University of Warwick, Coventry, CV4 7AL, UK. Email: \url{Charilaos.Efthymiou@warwick.ac.uk}}
\thanks{\hspace{0.38cm}$^\dagger$ University of Edinburgh, Edinburgh, EH8 9AB, UK. Email: \url{wfeng@ed.ac.uk}}

\maketitle

\thispagestyle{empty}

\begin{abstract}
We study the single-site Glauber dynamics for the fugacity $\lambda$,  Hard-core model  
on the random graph $G(n, d/n)$.
We show that  for the typical  instances  of the random graph $G(n,d/n)$ and for fugacity
$\lambda <  \frac{d^d}{(d-1)^{d+1}}$,    the mixing time of 
 Glauber dynamics is $n^{1 + O(1/\log \log n)}$.

Our result  improves on the recent elegant algorithm in [Bez{\'{a}}kov{\'{a}},  Galanis,  Goldberg  \v{S}tefankovi\v{c}; ICALP'22].
The algorithm there is a MCMC based sampling algorithm, but it is  not the Glauber dynamics.
Our algorithm here is {\em simpler}, as we use the classic Glauber dynamics. Furthermore, the bounds on mixing time 
we prove are smaller than those in Bez{\'{a}}kov{\'{a}} et al. paper,  hence our algorithm is also {\em faster}.

The main challenge in our proof is handling vertices with unbounded degrees. 
We provide stronger results with regard the spectral independence via branching values
and  show that the our Gibbs distributions satisfy the approximate tensorisation of the entropy.
We conjecture  that the bounds we have here are optimal for $G(n,d/n)$.

As  corollary of our analysis for the Hard-core model,  we also get bounds on the mixing time of the Glauber dynamics for 
the Monomer-dimer model on $G(n,d/n)$.  The bounds we get for this model are slightly better than those we have
for the Hard-core model

\end{abstract}

\newpage
\section{Introduction}


The  Hard-core model and the related problem of the geometry of independent sets on the spare 
random graph $G(n, d/n)$ is a fundamental area of study in discrete mathematics \cite{IndSetsGnpFrieze90,DaniMooreIS11}, 
in computer science they are studied in the context of the random {\em Constraint Satisfaction Problems}
\cite{CoEf15,GamSud14}, while in statistical physics they are studied as instances of {\em disordered systems}. 
Using the so-called {\em Cavity method} \cite{KrzakMRSZ07,BarKrzZdebZhan13}, physicists make some impressive  predictions about the 
independent sets of $G(n,d/n)$,  such as higher order replica symmetry breaking etc. 
Physicists' predictions are (typically) mathematically non-rigorous. Most of these predictions about independent sets
still remain open as basic natural objects in the study such as the
partition function, or the free energy are extremely  challenging to analyse.

The Hard-core model with fugacity $\lambda>0$, is a distribution over the {\em independent
sets} of an underlying graph $G$ such that every independent set $\sigma$ is assigned probability measure $\mu(\sigma)$ which is
proportional to $\lambda^{|\sigma|}$, where $|\sigma|$ is the cardinality of $\sigma$. 
Here, we consider  the case where the underlying graph is a typical  instance of the sparse random graph $G(n,d/n)$. 
This is  the  random graph on $n$  vertices, while each edge appears independently with probability $p=d/n$. 
The quantity $d>0$ corresponds to the {\em expected degree}.  For us here the expected degree is a bounded 
constant,   i.e., we have $d=\Theta(1)$, hence the graph is sparse.  

Our focus is  on approximate sampling  from the aforementioned distribution using 
{\em Glauber dynamics}. 
This is a classic, very popular, algorithm for approximate sampling. The  popularity of this process, mainly, 
is  due to its  simplicity and the  strong approximation  guarantees that provides. 
The efficiency of Glauber dynamics for sampling is studied by means of the {\em mixing time}.

Recently, there has been an ``explosion" of results about the mixing time of Glauber dynamics for 
{\em worst-case} instances the problem, e.g.  \cite{anari2020spectral, chen2020optimal, chen2022optimal,EftymiouTopological}.
Combined with the earlier hardness results in  \cite{sly2010computational,SS14,galanis2016inapproximability}
one could claim that for worst-case instances the behaviour of Glauber dynamics for the Hard-core model, but also the related 
approximate sampling-counting problem, is well understood. 
Specifically, for the graphs of maximum  degree $\Delta$,   Glauber dynamics  exhibits $O(n\log n)$ mixing  time for any 
fugacity $\lambda<(\Delta-1)^{\Delta-1}/(\Delta-2)^{\Delta}$, while  the hardness results  support that this  
region of $\lambda$ is best possible.

The aforementioned upper bound on $\lambda$ coincides with the {\em critical point}
 for the  uniqueness/non-uniqueness phase transition  of the Hard-core model on the infinite $\Delta$-regular 
tree  \cite{Kelly85}.
At this point in the discussion, perhaps, it is important to note the dependency of  the critical point
on the {\em maximum degree}.  This is the point where the situation with the random graph $G(n,d/n)$
differentiates from the worst case one.

For $G(n,d/n)$ and for the range of the expected degree $d$ we consider here,  typically,  almost all  of the vertices
in the graph, e.g., say 99\%,  are of degree very close to   $d$.  On the other hand, 
the maximum degree of $G(n,d/n)$  is as large as  $\Theta(\frac{\log n}{\log\log n})$, i.e., it is {\em unbounded}.
In light of this observation, it is natural to expect that the Glauber dynamics on the Hard-core model mixes fast 
for values of the fugacity that depend on the {\em expected degree}, rather the maximum degree.  Note that, 
this implies to use Glauber dynamics to sample from the Hard-core model with fugacity $\lambda$ 
taking {\em much larger} values than  what the worst-case bound implies. 

For $d>1$, let  $\lambda_c(d)=\frac{d^{d}}{(d-1)^{(d+1)}}$.  One of the main result in our paper is as follows: 
we show that for any $d>1$ and for typical instances of $G(n,d/n)$, the Glauber dynamics on
the Hard-core with any fugacity $\lambda <\lambda_c(d)$, exhibits mixing time which is 
$n^{1+\frac{C}{\log\log n} } = n^{1+o(1)}$, for some absolute constant $C>0$ which depends only 
on  $\lambda$ and $d$.  

It is our {\em conjecture} that the  bound on the mixing time is  tight. 
Furthermore, following intuitions 
from \cite{CoEf15}, as well as from {\em statistical physics} predictions in \cite{BarKrzZdebZhan13}, 
it is our {\em conjecture}  that the  bound $\lambda_c(d)$ on  the fugacity  $\lambda$ is also tight, in the following
sense:    for $\lambda>\lambda_c(d)$  it is not precluded that there is a region where efficient approximate 
sampling  is possible, however,   the approximation guarantees  are {\em weaker} than those we have here.

Our result improves on the elegant   sampling algorithm  that was  proposed recently 
in \cite{BGGS22}  for the same distribution,  i.e., the Hard-core model on $G(n,d/n)$. That algorithm, 
similarly to the one we consider here,  relies on the Markov Chain Monte Carlo method.
%
 The authors  use  Spectral Independence \cite{anari2020spectral,chen2020optimal} to show that the underlying Markov chain exhibits  mixing time  which is $O\left(n^{1+\theta }\right)$ for 
 any $\lambda<\lambda_{c}(d)$ and arbitrary small consant $\theta > 0$.  The idea that underlies 
the algorithm  in \cite{BGGS22}  is reminiscent of the {\em variable marking} technique that was introduced 
in \cite{MoitLLLCountJACM19} for approximate counting with the Lov\'{a}sz Local Lemma, and was 
further exploited  in \cite{FengGYZJACM21,FHY21,jain2021sampling,GalanisGGYApprCountCNF21}. 
Here, we use a different,  more straightforward,  approach and analyse directly the Glauber dynamics.

Note that both algorithms, i.e., here and in \cite{BGGS22},   allow for  the same range for the 
fugacity $\lambda$.  On the other hand, the algorithm we study here is the (much simpler) Glauber dynamics, 
while  the running time guarantees we obtain here are asymptotically better.

Previous works in the area, i.e., even before  \cite{BGGS22}, in order to prove their results and avoid the use of maximum degree, 
have been  focusing on various   parameters  of $G(n,d/n)$  such as the expected degree 
\cite{efthymiou2018sampling}, or the connective constant \cite{SSSY17}. Which, as it turns out are not that different
with each other.
Here,  we  utilise the notion of {\em branching value},  which is somehow related to the previous ones.

The notion of the branching value as well as its use for establishing Spectral Independence was introduced
 in \cite{BGGS22}. Unfortunately, the result there were not sufficiently strong to imply rapid mixing of Glauber dynamics.
Their analytic tools  for Spectral independence (and others) seems to not be able to handle all that well vertices with unbounded 
degree.  
Here we derive  stronger  results for Spectral independence than those in \cite{BGGS22}  in the sense that 
they are more {\em general}  and more {\em accurate}. Specifically, in our analysis we are able to accommodate 
vertices of {\em all degrees}, while we use a more elaborate matrix norm to  establish spectral independence,  
reminiscent of those introduced in \cite{EftymiouTopological}.  Furthermore, we  utilise results from 
\cite{chen2022optimal} that allow us deal with the unbounded degrees of the graph in order to establish 
our rapid mixing results.

\section{Results}
Consider the fixed graph $G=(V,E)$  on $n$ vertices.  Given the  parameter $\lambda>0$, which we call {\em fugacity}, we define 
the {\em Hard-core} model $\mu=\mu_{G,\lambda}$ to  be a distribution on the independent sets of the graph $G$,
Specifically,  every independent set $\sigma$ is assigned  probability measure $\mu(\sigma)$ defined by
\begin{align}
\mu(\sigma)\propto \lambda^{|\sigma|} \enspace, 
\end{align}
where $|\sigma|$ is equal to the size of the independent set $\sigma$.

We use $\{\pm 1\}^V$ to encode the configurations of the  Hard-core model, i.e., the independent sets of $G$. 
Particularly, the assignment $+1$  implies that the vertex is in the independent set, while $-1$ implies the opposite. 
We often use physics' terminology where  the vertices with assignment $+1$ are  called  ``occupied", whereas  the 
vertices with $-1$ are ``unoccupied".

We use the discrete time,   (single site)   {\em Glauber dynamics}   to approximately sample from  
the aforementioned distributions. 
%
  Glauber dynamics is a Markov chain with  state space the  support of the distribution $\mu$.
Typically, we assume that the chain   starts from an arbitrary configuration $X_0\in \{\pm 1\}^V$. For  
$t\geq 0$, the transition from the state $X_t$ to $X_{t+1}$ is according to the  following steps: 
\begin{enumerate}
\item Choose uniformly at random a vertex $v$. 
\item  For every vertex $w$ different than $v$, set $X_{t+1}(w)=X_t(w)$.
\item Set $X_{t+1}(v)$ according to the marginal of $\mu$ at $v$, conditional on 
the neighbours  of $v$ having the configuration  specified by $X_{t+1}$.
\end{enumerate}


It is standard that when a Markov chain satisfies a set of technical conditions called 
{\em ergodicity}, then it converges to a unique stationary distribution. For the cases 
we consider here,  Glauber dynamics is trivially ergodic, while the stationary distribution
is the corresponding Hard-core model  $\mu$.

Let  $P$ be the transition matrix of an  ergodic Markov chain $\{X_t\}$  with a finite state space 
$\Omega$ and equilibrium distribution $\mu$. For $t\geq 0$ and $\sigma\in \Omega$, let 
$P^{t}(\sigma, \cdot)$ denote the distribution of $X_t$ when the initial state of the chain 
satisfies $X_0=\sigma$.  The  {\em mixing time} of the Markov chain  $\{X_t\}_{t \geq 0}$ is 
defined by
\begin{align*}
   T_{\rm mix} &=  \max_{\sigma\in \Omega}\min { \left\{t > 0 \mid \Vert P^{t}(\sigma, \cdot) - \mu \Vert_{\rm TV} \leq \frac{1}{2 \mathrm{e}} \right \}}\enspace .
\end{align*}
Our focus is on  the mixing time of Glauber dynamics for the Hard-core model 
for the case where the underlying graph is a typical instance of $G(n,d/n)$, where the expected degree 
$d>0$ is a assumed to be a fixed number.


\subsection{Mixing Time for Hard-core Model}
%
%


For  $z>1$, we let the function $\lambda_c(z)=\frac{z^{z}}{(z-1)^{(z+1)}}$.  It is a well-known result from \cite{Kelly85} that 
the uniqueness region of the Hard-core model  on the $k$-ary tree, where $k\geq 2$,   holds for any  $\lambda$ such that
\begin{align}\nonumber
 \lambda< \lambda_c(k) \enspace.
\end{align}
The following theorem is the main result of this work.

\begin{theorem}\label{MainResultHC}
For fixed  $d>1$ and any $\lambda<\lambda_c(d)$, there is a constant $C>0$ such that the following is true:

Let $\mu_{\bG}$ be the Hard-core model with fugacity $\lambda$ on the graph $\bG\sim G(n,d/n)$.  With probability 
$1-o(1)$ over the instances of $\bG$, Glauber dynamics on $\mu_{\bG}$ exhibits mixing time 
\begin{align} \nonumber
T_{\rm mix}\leq n^{\left(1+\frac{C}{\log\log n} \right)} \enspace.
\end{align}
\end{theorem}


\subsection{Extensions to Monomer-dimer Model}
Utilising the techniques we develop in order to prove \Cref{MainResultHC}, we  
get mixing time bounds for the Glauber dynamics on the  Monomer-Dimer model on $G(n,d/n)$.  

Given a fixed graph $G=(V, E)$ and  a parameter $\lambda>0$, which we call {\em edge weight},
we define the Monomer-Dimer model $\mu=\mu_{G,\lambda}$ to be a distribution on the {\em matchings} 
of the graph $G$ such that every matching $\sigma$  is assigned probability measure $\mu(\sigma)$ defined by
\begin{align}
\mu(\sigma)\propto \lambda^{|\sigma|} \enspace, 
\end{align}
where $|\sigma|$ is equal to the number of edges in the matching  $\sigma$.

Note that the Hard-core model considers configurations on the vertices of $G$, while the Monomer-Dimer model
considers configurations on the edges.  Similarly to the independent sets, we use $\{\pm 1\}^E$  to encode the matchings
of $G$.  Specifically,   the assignment $+1$  on the edge $e$ implies that the edge is in matching, while $-1$ implies the opposite. 

For the Monomer-Dimer model the definition of Glauber dynamics $\Glauber$ extends in the natural way. That is, 
assume that the chain   starts from an arbitrary configuration $X_0\in ]{\pm 1}^E$. For  
$t\geq 0$, the transition from the state $X_t$ to $X_{t+1}$ is according to the  following steps: 
\begin{enumerate}
\item Choose uniformly at random an edge $e$. 
\item  For every edge $f$ different than $e$, set $X_{t+1}(f)=X_t(f)$.
\item Set $X_{t+1}(e)$ according to the marginal of $\mu$ at $e$, conditional on 
the neighbours  of $e$ having the configuration  specified by $X_{t+1}$.
\end{enumerate}

We consider the case of the Monomer-Dimer distribution where the underlying graph is an instance
of $G(n,d/n)$. We prove the following result.

\begin{theorem}\label{MainResultMD}
For fixed  $d>1$ and any $\lambda > 0$, there is a constant $C>0$ such that the following is true:

Let $\mu_{\bG}$ be the Monomer-dimer model with edge weight $\lambda$ on the graph $\bG\sim G(n,d/n)$.  With probability 
$1-o(1)$ over the instances of $\bG$, Glauber dynamics on $\mu_{\bG}$ exhibits mixing time 
\begin{align} \nonumber
T_{\rm mix}\leq n^{\left(1+C\sqrt{\frac{\log \log n}{\log n}} \right)} \enspace.
\end{align}
\end{theorem}
The proof of \Cref{MainResultMD} can be found in \Cref{sec:MonoDiProofs}. 

For the Monomer-dimer model on general graphs, the best-known result is the $\tilde{O}(n^2m)$ mixing time of the Jerrum-Sinclair chain~\cite{jerrum1989approximating}, where $m = |E|$ is the number of edges.
For graphs with bounded maximum degree $\Delta = O(1)$, the spectral independence technique proved the $O(n \log n)$ mixing time of Glauber dynamics~\cite{chen2020optimal}. 
However, this result cannot be applied directly to the random graph $G(n,d/n)$, because the maximum degree of a random graph is typically unbounded.
For the Monomer-dimer model on $G(n,d/n)$,  \cite{BGGS22} gave a sampling algorithm with running time $n^{1+\theta}$, where $\theta > 0$ is an arbitrarily small constant, 
and \cite{jain2021spectral} also proved the $n^{2+o(1)}$ mixing time of Glauber dynamics in a special case $\lambda = 1$.
Our result in \Cref{MainResultMD} proves the $n^{1+o(1)}$ mixing time of Glauber dynamics, which improves all the previous results for the Monomer-dimer model on the random graph $G(n,d/n)$ with constant $\lambda$.

We remark that for the Monomer-dimer model, we actually proved the $n^{1+o(1)}$ mixing time of Glauber dynamics on \emph{all} graphs satisfying $\Delta \log^2 \Delta = o(\log ^2 n)$. See \Cref{thmmatchingproof} for a more general result.

Note that, apart from \Cref{sec:MonoDiProofs}, the rest of the paper focuses on the Hard-core model, i.e.,  proving \Cref{{MainResultHC}}.

\subsubsection*{Notation} Suppose that we  are given a Gibbs distribution $\mu$ on the graph $G=(V,E)$.
We denote with $\Omega$  the support of $\mu$. 

Suppose that $\Omega$ is a set of configuration at the vertices of $G$.  Then, for any $\Lambda \subseteq V$ and 
any $\tau\in \{\pm 1\}^\Lambda$, we let $\mu^{\Lambda,\tau}$ (or $\mu^\tau$ if $\Lambda$ is clear from the context) denote the distribution $\mu$
conditional on that the configuration at $\Lambda$ is $\tau$. Alternatively, we use the notation 
$\mu(\cdot\ |\ (\Lambda,\tau))$ for the same conditional distribution.  We let $\Omega^{\tau}\subseteq \Omega$  
be  the support of $\mu^{\Lambda, \tau}$. We call $\tau$ {\em feasible} if $\Omega^\tau$ is nonempty. 

 For any subset $S \subseteq V$, let $\mu_S$ denote the  
marginal of  $\mu$  at $S$, while let $\Omega_S$  denote the support of $\mu_S$. 
%
In a natural way,  we define the conditional marginal.   That is, for $\Lambda \subseteq V\setminus S$ and 
$\sigma\in \{\pm 1\}^{\Lambda}$, we let $\mu^{\Lambda, \sigma}_{S}$  (or $\mu^\sigma_S$ if $\Lambda$ is clear 
from the context) denote the marginal at $S$ conditional on the configuration at $\Lambda$ being $\sigma$. Alternatively 
we use $\mu_S(\cdot\ |\  (\Lambda, \sigma))$ for $\mu^{\sigma}_{S}$. We let $\Omega^{\sigma}_S$ denote the support of $\mu^{\sigma}_{S}$.

All the above notation for configurations on the vertices of $G$ can be extended naturally for configurations on 
the edges of the graph $G$.  We omit presenting it, because it is very similar to the above. 

\subsection{ Hard-core Model - Entropy Tensorisation for Rapid Mixing}

We prove  \Cref{MainResultHC}  by exploiting the notion  of   {\em approximate tensorisation of the entropy}.  

Let $\mu$ be a distribution with support $\Omega \subseteq \{\pm 1\}^V$. For any function $f:\Omega\to\mathbb{R}_{\geq 0}$, 
we let $\mu(f) = \sum_{x \in \Omega}\mu(x)f(x)$, i.e., $\mu(f)$ is the expected value of $f$ with respect to $\mu$.
Define the entropy of $f$ with respect to $\mu$ by 
\begin{align}\nonumber
\entropy_{\mu}(f) &= \textstyle \mu\left(f\log\frac{f}{\mu(f)}\right) \enspace,
\end{align} 
where we use the convention that $0 \log 0 = 0$.

Let $\tau \in \Omega_{V \setminus S}$ for some $S\subset V$.
Define the function $f_\tau: \Omega^\tau_S \to \mathbb{R}_{\geq 0}$ by having 
 $f_\tau(\sigma) = f(\tau \cup \sigma)$ for all  $\sigma \in \Omega^\tau_S$ \footnote{ With a slight abuse of notation we use 
$\tau \cup \sigma$ to indicate the configuration what agrees with $\tau$ at $S$ and with $\sigma$ at $V\setminus S$.}.
Let $\entropy_S^\tau(f_\tau)$ denote the entropy of $f_\tau$ with respect to the conditional 
distribution $\mu^\tau_S$. Furthermore, we let 
\begin{align}\nonumber 
\mu(\entropy_S(f)) &= \sum_{\tau \in \Omega_{V \setminus S}} \mu_{V \setminus S}(\tau)\entropy_S^\tau(f_\tau) \enspace,
\end{align}
i.e.,  $\mu(\entropy_S(f))$ is the average of the entropy $\entropy_S^\tau(f_\tau)$ with respect to the measure 
$\mu_{V \setminus S}(\cdot)$.
When  $S=\{v\}$, i.e., the set $S$ is a singleton, we  abbreviate $\mu(\entropy_{\{v\}}(f))$ to $\mu(\entropy_v(f))$.

\begin{definition}[Approximate Tensorisation of Entropy]
A distribution $\mu$ with support $\Omega \subseteq \{\pm 1\}^V$ satisfies the approximate tensorisation of entropy with constant $C>0$ if for 
all $f:\Omega\to\mathbb{R}_{\geq 0}$ we have that
\begin{align}\nonumber
\entropy_{\mu}(f)\leq C\cdot \sum_{v\in V}\mu\left( \entropy_v(f)\right) \enspace.
\end{align}
\end{definition}

\noindent

On can establish bounds on the mixing time of Glauber dynamics by means of the  approximate tensorisation of entropy of 
the equilibrium distribution $\mu$.
Specifically, if $\mu$ satisfies the approximate tensorisation of entropy with constant $C$, then after every transition  of 
Glauber dynamics, the Kullback–Leibler divergence\footnote{ For discrete probability distributions 
$P$ and $ Q$ on a discrete space $\mathcal{X}$, the Kullback–Leibler divergence is defined by 
$D_{\rm KL}(P||Q)=\sum_{x\in\mathcal{X}}P(x)\log\frac{P(x)}{Q(x)}$.} between the current  distribution and the stationary distribution decays  by a factor which is at least  $(1-C/n)$, where $n = |V|$ is the number of variables.

As far as  the mixing time of Glauber dynamics is concerned, if a distribution $\mu$ satisfies the approximate tensorisation of entropy  with parameter $C$
then we have  following  well known relation (e.g. see \cite[Fact 3.5]{chen2020optimal}),
\begin{align}\label{eq:TmixVsTensC}
 T_{\rm mix} & \leq \left\lceil C n \tp{\log \log \frac{1}{\mu_{\min}} + \log(2)+2  } \right \rceil, & \text{ where } \mu_{\min}=\min_{x \in \Omega}\mu(x)  \enspace.
\end{align}

In light of the above,  \Cref{MainResultHC} follows as a corollary  from the following result.

\begin{theorem}[Hard-core Model Tensorisation]\label{thrm:EntropyTensGnpHC}
For any fixed  $d>1$ and any $\lambda<\lambda_c(d)$, there is a constant $A>0$ that depends only on $d$ and $\lambda$ such that the following is true:

Let $\mu_{\bG}$ be the Hard-core model with fugacity $\lambda$ on the graph $\bG\sim G(n,d/n)$.  With probability 
$1-o(1)$ over the instances of $\bG$,  $\mu_{\bG}$ satisfies the approximate tensorisation of entropy with parameter $n^{A /\log \log n}$.
\end{theorem}

\begin{proof}[Proof of \Cref{MainResultHC}]
\Cref{MainResultHC} follows  from  \Cref{thrm:EntropyTensGnpHC} and \eqref{eq:TmixVsTensC}.

Specifically,  plugging the result from \Cref{thrm:EntropyTensGnpHC} into \eqref{eq:TmixVsTensC} we get the following:
with  probability $1-o(1)$ over the instances of $\bG$ we have that
\begin{align*}
 T_{\rm mix} &\leq   n^{1 + \frac{A}{\log \log n}}\tp{ \log \log \frac{1}{\mu_{\min}} + \log(2)+2}   \\ 
 &\leq n^{1 + \frac{A}{\log \log n}}\tp{ \log \log \tp{1+\lambda+\lambda^{-1}}^n + \log(2)+2}  \\
 &= n^{1 + \frac{A}{\log \log n}}\tp{\log n + \log\log(1+\lambda+\lambda^{-1}) } \leq n^{1 + \frac{2A}{\log \log n}} \enspace.
\end{align*}
For the second derivation, we note that for the Hard-core distribution  $\mu = \mu_{\bG}$, we have that 
$\mu_{\min}$ is at least $\min\{1,\lambda^n\}/(1+\lambda)^n$, which implies that  $\mu_{\min}\geq (1+\lambda+\lambda^{-1})^{-n}$.

Note that \Cref{MainResultHC} follows from the above, by setting $C=2A$.
\end{proof}

\section{Our Approach \& Contributions}
In this section we describe our approach towards establishing our results. Our focus is on  the Hard-core model.

\subsection{Tensorisation and  Block-Factorisation of Entropy}\label{sec:TensorBFact}
We establish the tensorisation of the entropy, described in  \Cref{thrm:EntropyTensGnpHC}, 
by exploiting  the recently introduced  notion of {\em block factorisation of entropy} in \cite{CP20}.
Specifically, we build on the framework introduced in \cite{chen2020optimal} to relate the tensorisation 
and the block factorisation of the entropy.

The  framework in \cite{chen2020optimal} relies on the assumption that the maximum degree of the underlying graph is 
bounded. Otherwise, the results it implies are not  strong.   In our setting here,    a vanilla application of this approach would not be sufficient to give  the desirable  bounds 
on the  tensorisation constant due to  the fact that the  typical instances of $G(n,d/n)$ have unbounded  maximum degree.   
To this end, we employ  techniques from \cite{chen2022optimal}.

Given the graph $G=(V,E)$, and  the integer $\ell\geq 0$, we let $\binom{V}{\ell}$ denote all subsets $S \subseteq V$ with $|S| = \ell$.

\begin{definition}[$\ell$-block Factorisation of Entropy]
Let $\mu$ be a distribution over $\{\pm 1\}^V$ and $1 \leq \ell \leq |V| = n$ be an integer.
The distribution $\mu$  satisfies the $\ell$ block factorisation of entropy with parameter $C$ if for all $f: \Omega \to \mathbb{R}_{\geq 0}$ we have that
\begin{align}\label{eq:LBlockFactEntr}
\entropy_{\mu}(f) \leq \frac{C}{\binom{n}{\ell}} \sum_{S \in \binom{V}{\ell}} \mu \tp{\entropy_S(f)} \enspace.
\end{align}
\end{definition}

The notion of the $\ell$ block factorisation of entropy generalises that  of  the approximate tensorisation of entropy.
Specifically, a distribution that satisfies the $\ell=1$ block factorisation of entropy with parameter $C$, 
also  satisfies the approximate tensorisation of entropy with parameter $C/n$.  

As far as the Hard-core model on $G(n,d/n)$ is concerned, we show the following theorem, which is one of the main technical 
results in our paper.

\begin{theorem}\label{GnpHardcoreBlockFact}
For fixed   $d>1$ and any $0<\lambda < \lambda_c(d)$,   consider  $\bG\sim G(n,d/n)$ and  let $\mu_{\bG}$ be the Hard-core model on $\bG$ with fugacity $\lambda$.  
With probability $1-o(1)$ over the instances of $\bG$ the following is true:  There is a constant $K=K(d,\lambda)>0$,  
such that for 
\begin{align*}
    \frac{1}{\alpha} &= \textstyle K \frac{\log n}{\log \log n} \enspace , 
\end{align*}
for  any $1/\alpha \leq  \ell < n$, $\mu_{\bG}$ satisfies the $\ell$-block factorisation of entropy with parameter 
$C = (\frac{en}{\ell})^{1+1/\alpha}$.  
\end{theorem}

Let us have a high level overview of how we use the $\ell$-block factorisation and particularly \Cref{GnpHardcoreBlockFact} 
to establish our entropy tensorisation result in \Cref{thrm:EntropyTensGnpHC}.

Note that   \Cref{GnpHardcoreBlockFact} essentially implies the following: Suppose that $G=(V,E)$ is a {\em typical instance} of 
$G(n,d/n)$.  Then, the Hard-core  model $\mu$ on $G$, with fugacity $\lambda<\lambda_c(d)$,  is such that for 
any $f:\Omega\to\mathbb{R}_{>0}$ we have  
\begin{align}\label{eq:HighLevelEntTensStepA}
\entropy_{\mu}(f) \leq \tp{\frac{\mathrm{e}}{\theta}}^{1+1/\alpha} \frac{1}{\binom{n}{\ell}} \sum_{S \in \binom{V}{\ell}} \mu \tp{\entropy_S(f)} \enspace,
\end{align}
where $\ell=\lceil \theta n \rceil $ and $\theta \in (0,1)$ is a constant satisfying $\lceil \theta n \rceil \geq 1/\alpha=\Omega(\log n / \log \log n)$.

Let $G[S]$  be  the subgraph of $G$ that is induced by the vertices in the set $S$.
On   the RHS of~\eqref{eq:HighLevelEntTensStepA},  the entropy is evaluated with respect to conditional distributions $\mu^\tau_S$, 
which is the Hard-core model   on  the subgraph  $G[S]$ given the boundary condition $\tau$ on $V \setminus S$.  

We let $C(S)$ denote the set   of connected components in $G[S]$. With a slight  abuse of notation, we use  
$U \in C(S)$ to  denote the  set of vertices in the component $U$, as well. 
It is not hard to see that  the Hard-core model $\mu^\tau_S$, for $\tau \in \Omega_{V \setminus S}$, 
 factorises as a product distribution over Gibbs marginals at the components $U\in C(S)$, i.e., 
\begin{align} \nonumber 
\mu^{\tau}_S &= \ \bigotimes_{U\in C(S)} \mu^\tau_U \enspace.
\end{align}
We use the following result for  the factorisation of entropy on product  distributions~\cite{Cesi01,CMT15,chen2020optimal}.

\begin{lemma}[\text{\cite[Lemma 4.1]{chen2020optimal}}] \label{lemma-ent-product}
For any $S \subseteq V$, any $\tau \in \Omega_{V \setminus S}$, any $f: \Omega_S^\tau \to \mathbb{R}_{\geq 0}$,
\begin{align*}
    \entropy^\tau_S(f) \leq \sum_{U \in C(S)}\mu^\tau_S[\entropy_U(f)] \enspace.
\end{align*}
\end{lemma}

Combining  \Cref{lemma-ent-product} and \eqref{eq:HighLevelEntTensStepA} we get that 
\begin{align}\label{eq-E-S}
  \entropy_{\mu}(f) &\leq   
 \tp{\frac{\mathrm{e}}{\theta}}^{1+1/\alpha} { \textstyle  
 \mathbb{E}_{\bS \sim \binom{V}{\ell}} \left[ \sum_{U \in C(\bS)} \mu \tp{\entropy_U(f)} \right]
 } \enspace,
\end{align}
where   $\bS \sim \binom{V}{\ell}$  denotes   that $\bS$ is a uniformly random element from  $\binom{V}{\ell}$. 

The above step allows us to reduce the proof of approximate tensorisation to that of the components 
in  $C(\bS)$.  We choose the parameter $\ell = \lceil \theta n \rceil$ so that the connected components in $C(\bS)$ are typically
small. 

In light of the above,  \Cref{thrm:EntropyTensGnpHC} follows by  establishing  two results:
The first one is to derive a bound on the constant of the approximate tensorisation of entropy  for 
the  components of size  $k$ in $C(\bS)$, for each $k>0$. The second result is to derive tail bounds on  the size 
of the components in  $C(\bS)$ for $\bS \sim \binom{V}{\ell}$.

\begin{lemma}\label{corollary-bound-hardcore}
For any fixed $d>0$, for any $\lambda<\lambda_c(d)$, consider $\bG\sim G(n,d/n)$.
With probability $1-o(1)$ over the instances of 
$\bG$, the following is true:

For any $k\geq 1$ and $H\subseteq V$ such that $|H|=k$, the Hard-core model $\mu_{H}$ on $\bG[H]$ with fugacity $\lambda$
 satisfies the approximate tensorization of entropy with constant  
 \begin{align}\label{eq-AT-min-bound}
     \mathrm{AT}(k) \leq  \min\left\{ 2k^2 \tp{1+\lambda+{1}/{\lambda}}^{2k+2},\  3 \log\tp{1+\lambda+{1}/{\lambda}} \cdot ((1+\lambda)k)^{2+2\eta} \right\} \enspace,
 \end{align}
 where $\eta = B (\log n)^{1/r}$, while   $B=B(d,\lambda)$ and  $r=r(d) \in (1,2)$  are constants that depend on $d, \lambda$.
\end{lemma}

As far as size of the components  in $C(\bS)$ is concerned, we use the following  result  from \cite{BGGS22}.

\begin{lemma}[\text{\cite{BGGS22}}] \label{corollary-component}
Let $d > 1$ be a constant. There is a constant $L=L(d)$ such that 
the following holds with probability at least $1 - o(1)$ over the $\bG \sim G(n,d/n)$. 
Let $\bS \sim \binom{V}{\ell}$, while let $C_v\subseteq \bS$  be the
set of vertices that are in the same component as vertex $v$ in $\bG[\bS]$. For any integer  $k \geq \log n$,
it holds that
\begin{align*}
\Pr[|C_v|=k] \leq (2\mathrm{e})^{\mathrm{e}Lk} \tp{\frac{2\ell}{n}}^k \leq (2\mathrm{e})^{\mathrm{e}Lk} \tp{2\theta}^k \enspace.
\end{align*}
\end{lemma}

\Cref{thrm:EntropyTensGnpHC} follows by combining  \Cref{GnpHardcoreBlockFact}, with Lemmas \ref{corollary-component} and \ref{corollary-bound-hardcore}.
For a full proof of \Cref{thrm:EntropyTensGnpHC}, see \Cref{sec:thrm:EntropyTensGnpHC}.

\subsection{Spectral Independence with Branching Values}\label{sec:SIWithBranching}  
An important component in our analysis is to establish  Spectral Independence bounds 
for the Hard-core model  on typical instances of $G(n,d/n)$. 

For worst-case graph instances (i.e., non random), typically, we establish Spectral Independence for a region 
of the parameters of the Gibbs distribution which is expressed in terms of the maximum degree $\Delta$ 
of the underlying graph $G$. As far as $G(n,d/n)$ is concerned, the maximum degree does not seem to 
be the appropriate graph parameter to consider for this problem. 

Here,  we  utilise the notion of {branching value}. 
The notion of the branching value as well as its use for establishing Spectral Independence was introduced
 in \cite{BGGS22}. Unfortunately, the result there were not sufficiently strong to imply rapid mixing of Glauber dynamics.
Here we derive  stronger  results for Spectral independence than those in \cite{BGGS22}  in the sense that 
they are more {\em general}  and more {\em accurate}. Specifically, in our analysis we are able to accommodate 
vertices of {\em all degrees}, while we use a more elaborate matrix norm to  establish spectral independence,  
reminiscent of those introduced in \cite{EftymiouTopological}.  Furthermore, we  utilise results from 
\cite{chen2022optimal} that allow us deal with the unbounded degrees of the graph in order to establish 
our rapid mixing results.


Before getting to further details in our discussion,  let us first introduce some basic  notions.  We start with the 
{\em pairwise influence matrix} $\cI^{\Lambda,\tau}_G$  and the related notion of Spectral Independence. 
These notions were first introduced in   \cite{anari2020spectral}.
In this paper, we use the absolute version introduced in \cite{feng2021rapid}.

Consider a {\em fixed} graph $G=(V,E)$. Assume that we are given a Gibbs distribution $\mu$ on the 
configuration space $\{\pm 1\}^V$. 
We define the  pairwise influence matrix $\cI^{\Lambda,\tau}_{G}$ as follows: for a set of 
vertices  $\Lambda\subset V$ and a configuration $\tau$ at $\Lambda$, 
the matrix $\cI^{\Lambda,\tau}_{G}$ is  indexed by the vertices in $V\setminus \Lambda$, 
while  for any two vertices, different with each other  $v,w\in V\setminus \Lambda$, if $w$ can take both values $\pm 1$ given $\tau$, we have that  
\begin{align}\label{Overviewdef:InfluenceMatrix}
\cI^{\Lambda,\tau}_{G}(w,u)=  \lnorm \mu_{u }(\cdot \ |\ (\Lambda, \tau),  (\{w\}, +))- \mu_{u }(\cdot \ |\ (\Lambda, \tau),  (\{w\}, -)) \rnorm_{\rm TV}
\enspace; 
\end{align}
if $w$ can only take one value in $\pm 1$ given $\tau$, we have $\cI^{\Lambda,\tau}_{G}(w,u) = 0$.
Also, we have that $\cI^{\Lambda,\tau}_{G}(w,w)=0$ for all $w\in V\setminus \Lambda$. That is, 
the diagonal  of $\cI^{\Lambda,\tau}_{G}$ is  always zero. 

Recall that, above,   $\mu_{u }(\cdot  \ |\ (\Lambda, \tau),  (\{w\}, 1))$ is the Gibbs marginal  that  vertex $u$,
conditional that the configuration at $\Lambda$ is $\tau$ and  the configuration at  $w$ is $1$.  
We have  the analogous  for $\mu_{u }(\cdot  \ |\ (\Lambda, \tau),  (\{w\}, -1))$.

\begin{definition}[Spectral Independence]\label{OverviewDef:SpInMu}
For a real number $\eta>0$,  the Gibbs distribution $\mu_G$ on  $G=(V,E)$ is  $\eta$-spectrally
independent, if for every $0\leq k\leq |V|-2$, $\Lambda\subseteq V$ of size $k$ and $\tau\in \{\pm 1\}^\Lambda$
the spectral radius of $\cI^{\Lambda,\tau}_{G}$ satisfies  that  $\rho(\cI^{\Lambda,\tau}_{G})\leq \eta$.  
\end{definition}

We bound the spectral radius of $\cI^{\Lambda,\tau}_G$  by means of
matrix norms.  Specifically,  we use the following norm of $\cI^{\Lambda,\tau}_G$
\begin{align}\label{def:OfInfluenceNormA}
\textstyle \lnorm D^{-1}\cdot \cI^{\Lambda,\tau}_G\cdot D\rnorm_{\infty},
\end{align}
where $D$ is the diagonal matrix indexed by the vertices in $V\setminus \Lambda$ such that
\begin{align}\label{def:OfInfluenceNormB}
D(u,u) &=\begin{cases}
        \deg_G(v)^{1/\chi} &\text{if } \deg_G(v) \geq 1 \\
        1 &\text{if } \deg_G(v) = 0 \enspace,
\end{cases}
\end{align}
where the parameter $\chi$ is being specified later.  

Let $G=(V,E)$ be a fixed graph. For any vertex $v \in V$ and integer $\ell \geq 0$,  we use $N_{v,\ell}$ to denote the number 
of {simple paths} with $\ell + 1$ vertices that start from $v$ in graph $G$.  By definition, we have that $N_{v,0} = 1$.

\begin{definition}[$d$-branching value]\label{def:DBranchingVal}
Let $d \geq 1$ be a real number and $G=(V,E)$ be a graph. For any vertex $v \in V$, the $d$-branching value $S_v$ is defined by $\sum_{\ell \geq 0} N_{v,\ell}/d^\ell$. 
\end{definition}

We establish spectral independence results  that utilise the notion of $d$-branching value that are similar 
to the following one.

\begin{theorem}\label{Overview:lemma-SI}
Let $d > 1$ be a real number and $G=(V,E)$ be a graph. 
Let $\mu_{G}$  be the Hard-core model with fugacity   $\lambda <\lambda_c(d)$.
For any  $\alpha>0$ such that the $d$-branching value  $S_v \leq \alpha$ for all $v \in V$ the following is true: 
  $\mu_{G}$  is $\eta$-spectrally independent for
\begin{align*}
    \eta     &\leq  C_0 \cdot  \alpha^{1/r} \enspace, 
\end{align*}
 where   $C_0=C_0(d, \lambda)$, while the quantity  $r=r(d) \in (1,2)$ are constants. 
\end{theorem}

\Cref{Overview:lemma-SI} is a special case of a stronger result we obtain, i.e.,  \Cref{lemma-SI}.
Also, note that \Cref{Overview:lemma-SI}  is {\em not} necessarily about $G(n,d/n)$. As a matter of fact in order to use the above
result for $G(n,d/n)$ we need to establish  bounds on its branching value. To this end, we use the following result from \cite{BGGS22}.

\begin{lemma}[\text{\cite[Lemma~9]{BGGS22}}]\label{lemma-d-factor}
Let $d \geq 1$. For any fixed $d' > d$, with  probability $1-o(1)$ over $\bG\sim G(n, d/n)$,   the $d'$-branching factor of every vertex in $G$ is at most $\log n$. 
\end{lemma}

It is worth mentioning that \Cref{lemma-d-factor}, here,  is a {\em weaker} version of Lemma~9 in \cite{BGGS22}, i.e., we do not really need the
full strength of the result there. 

Concluding this short introductory section about Spectral Independence, let us remark that for our results we 
work with the so-called   {\em Complete Spectral Independence} for the Hard-core model,  introduced in 
\cite{chen2021rapid,chen2022optimal}. This is more general a notion compared to the (standard)  Spectral Independence.
For further discussion see \Cref{sec:SIResults}.

\section{Entropy Factorisation from Stability and Spectral Independence}\label{sec:EntropyFactor}

In this section we   establish the   $\ell$-block factorisation of entropy  for the Hard-core model on $G(n,d/n)$
as it is described in \Cref{GnpHardcoreBlockFact}. To this end, we employ  techniques from \cite{chen2022optimal}.
This means that we study the Hard-core model on $G(n,d/n)$  in terms of the stability of ratios of the marginals 
and the so-called Complete  Spectral Independence.

\subsection{Ratios of Gibbs Marginals \&  Stability}\label{sec:StabilityMargRatio}
Consider the  {\em fixed} graph $G=(V, E)$ and a Gibbs distribution $\mu$ on this graph. 
For a vertex $w\in V$,  the region  $K \subseteq V \setminus\{w\}$ and   
$\tau\in  \{\pm 1\}^{K}$, we consider the {\em ratio of  marginals} at $w$ denoted as $R^{K, \tau}(w)$  such that
\begin{align}
R^{K, \tau}_{G}(w)=\frac{\mu_w(+1\ |\ K,  \tau )}{\mu_w(-1\ |\ K,  \tau)} \enspace.
\end{align}

Recall that  $\mu_w(\cdot \ |\ K,  \tau )$ denotes the marginal of the Gibbs distribution $\mu(\cdot \ |\ K,  \tau )$ 
at vertex $w$. Also, note that the  above  allows  for $R^{K, \tau}(w)=\infty$,  e.g.,  when  $\mu_w(-1\ |\ K,  \tau)=0$ 
and  $\mu_w(+1\ |\ K,  \tau )\neq 0$.

\begin{definition}[Marginal stability]
Let $\zeta > 0$ be a real number. The Gibbs distribution $\mu_G$ on $G=(V,E)$ is called  $\zeta$-marginally stable 
if for any $w\in V$, for any $\Lambda\subset V$, for any configuration $\sigma$ at $\Lambda$ and   any
$S\subseteq \Lambda$ we have that
\begin{align}
R^{\Lambda,\tau}_G(w) &\leq \zeta &  \textrm{and} &&
R^{\Lambda,\tau}_G(w) &\leq \zeta  \cdot R^{S,\tau_S}_G(w) \enspace. 
\end{align}
\end{definition}

As far as the stability of the Hard-core marginals at $G(n,d/n)$ is concerned, we prove the following result.

\begin{theorem}[Stability  Hard-Core Model]\label{Bounds4MarginalRatios}
For any fixed $d>0$, for any $\lambda<\lambda_c(d)$, consider $\bG\sim G(n,d/n)$ and let $\mu_{\bG}$
be the Hard-core model on $\bG$ with fugacity $\lambda$. 
With probability $1-o(1)$ over the instances $\bG$, $\mu_{\bG}$ is $2(1+\lambda)^{\frac{2\log n}{\log \log n} }$-marginally stable.
\end{theorem}

\begin{proof}
Let $\zeta=2(1+\lambda)^{2 \frac{\log n}{\log \log n} }$. Also, let   $N(w)$ be the set of the neighbours of $w$.

For any $\Lambda \subseteq V$ and any $\tau\in \{\pm\}^{\Lambda}$,  we have that   
$\mu_w(+1\ |\ \Lambda,  \tau )\leq \frac{\lambda}{1+\lambda}$.  One can see that the equality holds  
if $N(w)\subseteq \Lambda$ and for every $u\in N(w)$ we have that $\tau(w)=-1$. Noting that 
 $R^{\Lambda,\tau}_{\bG}(w)$ is increasing in the value of the Gibbs marginal $\mu_w(+1\ |\ \Lambda,  \tau)$, 
 it is immediate that
 \begin{align}\label{eq:UpperBound4RS}
\textstyle \Pr \left[ R^{\Lambda,\tau}_{\bG}(w) \leq {\lambda}<\zeta \quad \forall \Lambda\subseteq V \right] &=1  \enspace. 
 \end{align}
It remains to show  that  
\begin{align}\label{eq:LowerBound4RS}
\textstyle \Pr\left [R^{\Lambda,\tau}_{\bG}(w) \leq \zeta  \cdot R^{S,\tau_S}_{\bG}(w) \quad \forall \Lambda \subset V, \ \forall S\subset \Lambda \right] &=1-o(1)\enspace.
\end{align}
In light of \eqref{eq:UpperBound4RS}, \eqref{eq:LowerBound4RS} follows by showing that
\begin{align}\label{eq:Base4LowerBound4RS}
\textstyle \Pr\left [  R^{S,\tau_S}_{\bG}> 2\lambda \left( 1+\lambda\right)^{-2\frac{\log n}{\log\log n}} \quad \forall \Lambda \subset V, \ \forall S\subset \Lambda \right] &=1-o(1)\enspace.
\end{align}

If there is  $u\in N(w)$ such that $\tau(u)= + 1$, then $R^{\Lambda,\tau}_{\bG}(w) = 0$ and 
\eqref{eq:LowerBound4RS} holds trivially since $R^{S,\tau_S}_{\bG}(w)\geq 0$.
We focus on the case that all vertices $u \in N(w) \cap \Lambda$ satisfy  $\tau(u) = -1$.

Let  $\cE$ be the event that none of the  vertices in $N(w)$ is occupied, while let $\gamma_{S}$
be the probability of the event $\cE$ under the Gibbs distribution $\mu(\cdot \ |\ S, \tau_S)$.
It is standard to show that 
\begin{align*}
R^{S,\tau_S}_{\bG} &= \frac{\frac{\lambda}{1+\lambda} \gamma_{S}}{1-\frac{\lambda}{1+\lambda}\gamma_{S}} \enspace.
\end{align*}
Noting that the function  $f(x)=\frac{x}{1-x}$ is increasing in $x\in (0,1)$, while $\gamma_{S}\geq (\frac{1}{1+\lambda})^{\deg_{\bG}(w)}$, we have that
\begin{align*}
 R^{S,\tau_S}_{\bG} &\geq \frac{\frac{\lambda}{1+\lambda}(\frac{1}{1+\lambda})^{\deg_{\bG}(w)}}{1-\frac{\lambda}{1+\lambda}(\frac{1}{1+\lambda})^{\deg_{\bG}(w)}} 
 = \frac{\lambda}{({1+\lambda})^{\deg_{\bG}(w)+1}-\lambda}  \enspace.
\end{align*}
From the above it is immediate to get \eqref{eq:Base4LowerBound4RS}. Specifically, it follows from the above inequality and 
  \Cref{lemma-MaxDegGnp} which implies that  for any  fixed number $\epsilon>0$,  the maximum degree in  $\bG$ is  less 
  than $(1+\epsilon)\frac{\log n}{\log\log n}$  with probability $1-o(1)$ .  

This  concludes the proof of \Cref{Bounds4MarginalRatios}.
\end{proof}

\subsection{(Complete) Spectral Independence}\label{sec:SIResults} 

\noindent
\noindent
The notions of   the   pairwise influence matrix  $\cI^{\Lambda,\tau}_G$ and the  Spectral Independence, 
as we introduce them in  \Cref{sec:SIWithBranching}, are typically used to establish bounds on
the spectral gap for Glauber dynamics and hence derive bounds on the mixing time of the chain. 

The authors in \cite{chen2020optimal}, make a further use  of Spectral Independence to obtain the approximate tensorisation of entropy.
Unfortunately,  a vanilla application of their technique is not sufficient to prove our tensorisation results, mainly, because of the unbounded 
degrees we typically have  in  $G(n,d/n)$.

In this work, we exploit ideas from \cite{chen2020optimal}  together with the related  notion of the 
{\em Complete  Spectral Independence},
in order to establish our factorisation results for the entropy in \Cref{GnpHardcoreBlockFact}. 
Specifically, we utilise  the connection between complete spectral independence and the $\ell$ block factorisation 
of entropy that  was established in \cite{chen2022optimal} (see further details in the following section).

Since the notions of the   pairwise influence matrix  $\cI^{\Lambda,\tau}_G$ and the  Spectral Independence 
are so important, let us recall them  once more, even though they have already been defined in \Cref{sec:SIWithBranching}.
Consider a {\em fixed} graph $G=(V,E)$. Assume that we are given a Gibbs distribution $\mu$ on the 
configuration space $\{\pm 1\}^V$. 

We define the  pairwise influence matrix $\cI^{\Lambda,\tau}_{G}$ as follows: for a set of 
vertices  $\Lambda\subset V$ and a configuration $\tau$ at $\Lambda$, 
the matrix $\cI^{\Lambda,\tau}_{G}$ is  indexed by the vertices in $V\setminus \Lambda$, 
while  for any two vertices $v,w\in V\setminus \Lambda$, different with each other,  if $w$ can take both values $\pm 1$ given $\tau$, we have that  
\begin{align}\label{def:InfluenceMatrix}
\cI^{\Lambda,\tau}_{G}(w,u)=  \lnorm \mu_{u }(\cdot \ |\ (\Lambda, \tau),  (\{w\}, +))- \mu_{u }(\cdot \ |\ (\Lambda, \tau),  (\{w\}, -)) \rnorm_{\rm TV}
\enspace; 
\end{align} 
if $w$ can only take one value in $\pm 1$ given $\tau$, we have $\cI^{\Lambda,\tau}_{G}(w,u) = 0$.
Also, we have that $\cI^{\Lambda,\tau}_{G}(w,w)=0$ for all $w\in V\setminus \Lambda$. That is, 
the diagonal  of $\cI^{\Lambda,\tau}_{G}$ is  always zero.

Recall that, above,   $\mu_{u }(\cdot  \ |\ (\Lambda, \tau),  (\{w\}, 1))$ is the Gibbs marginal  that  vertex $u$,
conditional that the configuration at $\Lambda$ is $\tau$ and  the configuration at  $w$ is $1$.  
We have  the analogous  for $\mu_{u }(\cdot  \ |\ (\Lambda, \tau),  (\{w\}, -1))$.

\begin{definition}[Spectral Independence]\label{Def:SpInMu}
For a real number $\eta>0$,  the Gibbs distribution $\mu_G$ on  $G=(V,E)$ is  $\eta$-spectrally
independent, if for every $0\leq k\leq |V|-2$, $\Lambda\subseteq V$ of size $k$ and $\tau\in \{\pm 1\}^\Lambda$
the spectral radius of $\cI^{\Lambda,\tau}_{G}$ satisfies  that  $\rho(\cI^{\Lambda,\tau}_{G})\leq \eta$.  
\end{definition}

We proceed to  introduce the  Complete Spectral Independence. First, consider the notion of the  Magnetising operation.

\begin{definition}[Magnetising operation]\label{def:Magnetising}
Let $\mu_G$ be a Gibbs distribution on the graph $G=(V,E)$.
For any  local fields $ \vec{\phi} \in \mathbb{R}_{>0}^V$, the magnetised distribution $\vec{\phi} * \mu$ satisfies 
\begin{align*}
\forall \sigma \in \{\pm 1\}^V,\quad (\vec{\phi} * \mu)(\sigma) \propto \mu(\sigma)\prod_{v \in V:\sigma_v = +1}\phi_v \enspace.
\end{align*}
We denote $\vec{\phi} * \mu$ by $\phi * \mu$ if $\vec{\phi}$ is a constant vector with value $\phi$.
\end{definition}

Suppose that  $\mu$ is the Hard-core model on $G$ with fugacity $\lambda$.
It is immediate that the  magnetisied  distribution  $\vec{\phi} * \mu$  can be viewed as  
the {\em non-homogenious} Hard-core model  such that   each vertex $v$ has its own 
fugacity $\lambda_v=\lambda\cdot \phi_v$.

\begin{definition}[Complete Spectral Independence]\label{Def:ComSpInMu}
For two reals $\eta>0$ and $\xi > 0$,  the Gibbs distribution $\mu_G$ on  $G=(V,E)$ is  $(\eta,\xi)$-completely spectrally
independent, if the magnetised distribution $\vec{\phi} * \mu$ is $\eta$-spectrally independent for all $\vec{\phi} \in (0,1+\xi]^V$.
\end{definition}

As far as the Hard-core model on the random graph $G(n,d/n)$ is concerned, we prove the following result.

\begin{theorem}\label{SI4HardCore} 
For any  fixed $d > 1$ and $\lambda < \lambda_c(d)$, there exist bounded constants $r = r(d,\lambda) \in (1,2)$, 
$B = B(d,\lambda) > 0$ and $s = s(d,\lambda) > 0$ such that the following  holds:

Consider  $\bG\sim G(n,d/n)$ and  let $\mu_{\bG}$ be the Hard-core model on $\bG$ with fugacity $\lambda$. 
With probability $1-o(1)$ over the instances of $\bG$, $\mu_{\bG}$ is  $(B\cdot (\log n)^{1/r},s)$-completely 
spectrally independent.
\end{theorem}
The proof of  \Cref{SI4HardCore} appears in Section \ref{sec:SI4HardCore}.

\subsection{Entropy Block Factorisation - Proof of {\Cref{GnpHardcoreBlockFact}}}

The following theorem,  from \cite{chen2022optimal},   allows us to derive a bound on the 
$\ell$- block factorisation parameter  of the entropy  by using  the result in
\Cref{Bounds4MarginalRatios} for the stability of Gibbs marginals   and the  
 result   in  \Cref{SI4HardCore} for Complete Spectral Independence.

\begin{theorem}[\text{\cite[Lemma 2.3]{chen2022optimal}}]\label{thm-EI}
 Let $\eta > 0, \xi > 0$ and $\zeta > 0$ be parameters.
 Let $\mu_G$ be a Gibbs distribution on $G=(V,E)$.
 If $\mu_G$ is $(\eta,\xi)$-completely spectrally independent and $\zeta$-marginally stable, then for any $1/\alpha \leq \ell < n$, $\mu_G$ satisfies the $\ell$ block factorisation of entropy with parameter $C = (\frac{en}{\ell})^{1+1/\alpha}$, where 
 \begin{align*}
     \alpha = \min \left\{ \frac{1}{2\eta}, \frac{\log(1+\xi)}{\log(1+\xi) + \log 2\zeta} \right\} \enspace.
 \end{align*}
\end{theorem}

\begin{proof}[Proof of \Cref{GnpHardcoreBlockFact}]
From \Cref{SI4HardCore} we have the following: with probability $1-o(1)$ over the instances of $\bG$ we have that $\mu_{\bG}$
is $(\eta,s)$-completely spectrally independent where $s=s(d,\lambda)$ is  {\em  constant}, while 
\begin{align}\nonumber 
\eta &= B \cdot (\log n)^{1/r} =o\left(\frac{\log n}{\log\log n}\right),
\end{align}
where $B=B(d,\lambda)$ and $r = r (d, \lambda) \in (1,2)$ are constants specified  in the statement of \Cref{SI4HardCore}.
The second equality above follows by noting that $1/r <1$,   bounded away from   $1$.

Furthermore, from \Cref{Bounds4MarginalRatios} we have the following: With probability $1-o(1)$ over the instances of 
$\bG$, the distribution $\mu_{\bG}$ is $\zeta$-marginally stable, where 
\begin{align}\nonumber 
\zeta\leq 2(1+\lambda)^{2\frac{\log n}{\log\log n}} \enspace.
\end{align}

In light of all the above, the theorem follows by plugging the above values into  \Cref{thm-EI}. 
\end{proof}

\section{ Approximate Tensorisation of Entropy}\label{sec:thrm:EntropyTensGnpHC}

In this section we prove our results related to the approximate tensorisation of the entropy. 
These are \Cref{thrm:EntropyTensGnpHC} and  \Cref{corollary-bound-hardcore}.

\subsection{Proof of  \Cref{thrm:EntropyTensGnpHC}}
In this section we give the full proof of \Cref{thrm:EntropyTensGnpHC}.
Recall  the high level description of the steps we follow towards this  endeavour   in \Cref{sec:TensorBFact}.

\begin{proof}[Proof of \Cref{thrm:EntropyTensGnpHC}]
From  \Cref{GnpHardcoreBlockFact}  we have the following:  For $d>1$ and $\lambda<\lambda_c(d)$, 
consider $\bG\sim G(n,d/n)$, while let $\mu=\mu_{\bG}$ be the Hard-core model on $\bG$ with fugacity $\lambda$. 
Let the  number $\theta=\theta(d,\lambda)$ in the interval $(0,1)$  be a parameter 
whose value is going to be  specified later. Then,  with probability $1-o(1)$ 
over the instances of $\bG$, for $\ell=\lceil \theta n \rceil $ and for  any $f:\Omega\to\mathbb{R}_{>0}$ 
we have that
\begin{align}\label{eq:EntTensStepA}
\entropy_{\mu}(f) \leq \tp{\frac{\mathrm{e}}{\theta}}^{1+1/\alpha} \frac{1}{\binom{n}{\ell}} \sum_{S \in \binom{V}{\ell}} \mu \tp{\entropy_S(f)}.
\end{align}

Recall that  $C(S)$ denotes the set   of connected components in $\bG[S]$, the subgraph that is induced by  vertices in $S$. 
With a slight  abuse of notation, we use   $U \in C(S)$ to  denote the  set of vertices in the component $U$. 
By the conditional independence property of the Gibbs distribution and \Cref{lemma-ent-product}, we have
\begin{align}
   \entropy_{\mu}(f) &\leq  \tp{\frac{\mathrm{e}}{\theta}}^{1+1/\alpha} \frac{1}{\binom{n}{\ell}}  \sum_{S \in \binom{V}{\ell}} \sum_{U \in C(S)} \mu \tp{\entropy_U(f)} \nonumber\\
 \tp{\text{by~\Cref{{corollary-bound-hardcore}} }}\quad  &\leq \tp{\frac{\mathrm{e}}{\theta}}^{1+1/\alpha} \frac{1}{\binom{n}{\ell}}  \sum_{S \in \binom{V}{\ell}} \sum_{U \in C(S)}  \mathrm{AT}(|U|)\sum_{v \in U}\mu[\entropy_v(f)] \nonumber\\
 &\leq  \tp{\frac{\mathrm{e}}{\theta}}^{1+1/\alpha}  \sum_{v \in V}\mu[\entropy_v(f)] \sum_{k \geq 1} \mathrm{AT}(k) \Pr[|C_v| = k] \label{eq:EntrTensBase} \enspace,
\end{align}
where $C_v$ is the connected component in $\bG[S]$, where $S$ is sampled from $\binom{V}{\ell}$ uniformly at random.
In order to bound the innermost summation on the R.H.S. of \eqref{eq:EntrTensBase} we distinguish two cases for $k$.
For $1 \leq k \leq \log n$, we use the trivial bound $\Pr[|C_v| = k] \leq 1$, while   \Cref{corollary-bound-hardcore} implies that 
\begin{align*}
 \sum_{k = 1}^{\log n} \mathrm{AT}(k) \Pr[|C_v| = k] &\leq  \sum_{k = 1}^{\log n} \mathrm{AT}(k) = \sum_{k = 1}^{\log n} 3 \log\tp{1+\lambda+\lambda^{-1}} \cdot ((1+\lambda)k)^{2+2\eta} \\
 &\leq 3 \log\tp{1+\lambda+\lambda^{-1}} \cdot  \log n \cdot ((1+\lambda)\log n)^{2+2\eta}\\
 &\leq 3 \log\tp{1+\lambda+\lambda^{-1}} \cdot ((1+\lambda)\log n)^{3+2\eta} \enspace , 
\end{align*}
where $\eta = B (\log n)^{1/r}$, for  constants  $B=B(d,\lambda)$ and  $r=r(d) \in (1,2)$.
Elementary calculations imply that 
\begin{align}\label{eq:EntrTensSmall-K}
 \sum_{k = 1}^{\log n} \mathrm{AT}(k) \Pr[|C_v| = k]  &\leq 3 \log\tp{1+\lambda+\lambda^{-1}} \cdot  ((1+\lambda)\log n)^{3+2\eta}  \leq n^x \enspace,
\end{align}
for    $ x =o\tp{\frac{1}{\log \log n}}$.

For $k \geq \log n$, we use the bound in~\Cref{corollary-component} for $\Pr[|C_v| = k]$, while from  \Cref{corollary-bound-hardcore} we have 
\begin{align*}
\sum_{k \geq \log n} \mathrm{AT}(k) \Pr[|C_v| = k] \leq  2k^2 \tp{1+\lambda+\lambda^{-1}}^{2k+2} (2\mathrm{e})^{\mathrm{e}Lk} (2\theta)^k \enspace,  
\end{align*}
where $L = L(d)$ is the parameter in \Cref{corollary-component}. 
We  choose  sufficiently small $\theta = \theta(d,\lambda)$  such that
\begin{align*}
\forall k \geq 1, \quad   2k^2 \tp{1+\lambda+\lambda^{-1}}^{2k+2} (2\mathrm{e})^{\mathrm{e}Lk} (2\theta)^k \leq \tp{{1}/{2}}^k \enspace.
\end{align*}
This implies that
\begin{align}\label{eq:EntrTensLarge-K}
   \sum_{k \geq \log n} \mathrm{AT}(k) \Pr[|C_v| = k] \leq  \sum_{k \geq \log n}\tp{\frac{1}{2}}^k \leq 1 \enspace.
\end{align}
Plugging  \eqref{eq:EntrTensSmall-K}, \eqref{eq:EntrTensLarge-K} into \eqref{eq:EntrTensBase}, we get the following:
With probability  $1-o(1)$ over the instances of $\bG$ we have that 
\begin{align*}
  \entropy_{\mu}(f)  \leq     \tp{\frac{\mathrm{e}}{\theta}}^{1+1/\alpha}\tp{n^{\tp{\frac{1}{\log \log n}}} + 1}  \sum_{v \in V}\mu[\entropy_v(f)] \enspace.
\end{align*}
Since, by \Cref{GnpHardcoreBlockFact} we have that  $\frac{1}{\alpha} =K (\frac{\log n}{\log \log n})$, for a constant $K=K(d,\lambda)$, 
and $\theta = \theta(d,\lambda)$ is also a constant, the above inequality can be written  as follows:  there is a constant $A=A(d,\lambda)$ such that 
\begin{align*}
     \entropy_{\mu}(f)  \leq n^{\tp{\frac{A}{\log \log n}} } \sum_{v \in V}\mu[\entropy_v(f)] \enspace .
\end{align*}
The above concludes the proof of \Cref{thrm:EntropyTensGnpHC}. 
\end{proof}

\subsection{Proof of \Cref{corollary-bound-hardcore}}\label{sec:corollary-bound-hardcore}

 \Cref{corollary-bound-hardcore} follows as a corollary from  the following result.

\begin{lemma}\label{lemma-crude-bound}
Let $\lambda>0$,  let the graph $G=(V,E)$, while let $\mu_G$ be the Hard-core model on $G$ with fugacity $\lambda$.
Let $M \subseteq V$ be a subset of vertices and $\sigma \in \{\pm1\}^{V \setminus M}$ on $V \setminus M$.
The conditional distribution $\mu_{M}(\cdot\ |\ V \setminus M,\sigma)$ satisfies the approximate tensorisation of entropy with constant
\begin{align}\label{eq:lemma-crude-boundA}
C = 2|M|^2 \tp{1+\lambda+{1}/{\lambda}}^{2|M|+2} \enspace.
\end{align}
Furthermore, if $\mu$ is $\eta$-spectrally independent, for some number $\eta$, then 
$\mu_{M}(\cdot\ |\ V \setminus M,\sigma)$ satisfies the approximate tensorisation of entropy with constant
\begin{align}\label{eq:lemma-crude-boundB}
C = 3 \log\tp{1+\lambda+{1}/{\lambda}} \cdot ((1+\lambda)|M|)^{2+2\eta} \enspace.
\end{align}
\end{lemma}

In light of \Cref{lemma-crude-bound} the proof of \Cref{corollary-bound-hardcore}  is straightforward. 
 In what follows, we provide its proof for the sake of completeness. 

\begin{proof}[Proof of \Cref{corollary-bound-hardcore}]
The first bound in \eqref{eq-AT-min-bound} follows directly from the first part of \Cref{lemma-crude-bound}.
We now prove the second bound in \eqref{eq-AT-min-bound}.

By \Cref{SI4HardCore}, with  probability $1-o(1)$ we have that $\mu_{\bG}$ is $( B(\log n)^{1/r}, s)$-completely spectrally independent, 
where $r = r(d,\lambda) \in (1,2)$, $B = B(d,\lambda) > 0$ and $s = s(d,\lambda) > 0$. As a consequence,  
$\mu_{\bG}$ is $B(\log n)^{1/r}$-spectrally independent with probability $1-o(1)$. Hence, the second bound in \eqref{eq-AT-min-bound} 
follows from the second part of \Cref{lemma-crude-bound}.
\end{proof}

\begin{proof}[Proof of \Cref{lemma-crude-bound}]
Note that this result is for a fixed graph $G$.  Also, since we talk about tensorisation of entropy, w.l.o.g. assume that $|M|\geq 2$.

For brevity, we use $\pi$ to denote the distribution $\mu_{M}(\cdot\ |\ V \setminus M,\sigma)$. 
We use $\Omega \subseteq \{\pm1\}^M$ to denote the support of the  distribution $\pi$. 
Let $P: \Omega \times \Omega \to \mathbb{R}_{\geq 0}$ denote the transition matrix of the Glauber dynamics on $\pi$.
It is elementary to verify that the Glauber dynamics on $\Omega$ is ergodic. 

The Poincar\'{e} inequality for $\pi$ is that  
\begin{align*}
    \forall f: \Omega \to \mathbb{R}_{\geq 0}:\quad \gamma \var_{\pi}[f] \leq \frac{1}{|M|}\sum_{v \in \Lambda}\pi[\var_v(f)] \enspace, 
\end{align*}
where $\gamma$ is the {\em Poincar\'{e} constant}, a.k.a. spectral gap of $P$. The log-Sobolev inequality for $\pi$ is that  
\begin{align*}
     \forall f: \Omega \to \mathbb{R}_{\geq 0}:\quad \alpha \entropy_{\pi}[f] \leq \frac{1}{|M|}\sum_{v \in \Lambda}\pi[\var_v(\sqrt{f})] \enspace, 
\end{align*}
where $\alpha$ is the {\em log-Sobolev constant}. 

The following  well-known inequality from 
\cite[Proposition 1.1]{CMT15} relates  $C$, the approximate tensorisation constant we want to bound,  and the log-Sobolev 
constant $\alpha$.  We have that
\begin{align*}
C \leq \frac{1}{\alpha |M|} \enspace.
\end{align*}
From \cite[Corollary A.4]{DS96}  we have the following relation between $\gamma$ and $\alpha$
\begin{align*}
\alpha &\geq \frac{1-2\pi_{\min}}{\log(1/\pi_{\min} -1)} \cdot \gamma, & &\text{where } \pi_{\min} = \min_{\sigma \in \Omega}\pi(\sigma) \enspace.  
\end{align*}
We may assume $|\Omega| \geq 3$,  as $|\Omega| \leq 2$ is the trivial case in which $C \leq 1$.
It holds that  $\pi_{\min} \leq \frac{1}{3}$.
Combining the above two inequalities together, we have
\begin{align}\label{eq-C-bound}
   C \leq  \frac{3 \log(1/\pi_{\min})}{|M| \gamma} \enspace.
\end{align}
Next, we use Cheeger's inequality to derive a crude lower bound on the spectral gap $\gamma$. 
Recall that Cheeger's inequality implies that  
\begin{align}\label{eq:CheegersIneq}
\gamma \geq {\Phi^2}/{2} \enspace,
\end{align}
where
\begin{align*}
    \Phi = \min_{\Omega_0 \subseteq \Omega: \pi(\Omega_0) \leq \frac{1}{2}} \frac{1}{\pi(\Omega_0)} \sum_{x \in \Omega_0} \sum_{y \in \Omega \setminus \Omega_0} \pi(x) P(x,y) \enspace,
\end{align*}
and (as defined above)  $P$ is the transition matrix of the Glauber dynamics. 

Since the Glauber dynamics  is ergodic,   for any $\Omega_0 \subset \Omega$, there exists 
$x \in \Omega_0$ and $y \in \Omega \setminus \Omega_0$ such that $P(x,y) > 0$. By the definition of the Glauber 
dynamics, such $x$ and $y$ satisfy  $P(x,y) \geq \frac{1}{|\Lambda|}\min\{\frac{1}{1+\lambda}, \frac{\lambda}{1+\lambda}\}$. 
Hence, we have that 
\begin{align}\label{eq:CheegConstBound}
 \Phi &\geq \frac{2 \pi_{\min}}{| M |} \min\left\{\frac{1}{1+\lambda}, \frac{\lambda}{1+\lambda}\right\} \enspace. 
\end{align}
Plugging \eqref{eq:CheegersIneq} and \eqref{eq:CheegConstBound} into  \eqref{eq-C-bound}, we get that
\begin{align*}
   C &\leq  \frac32 \frac{ |M|\log(1/\pi_{\min})}{\left(\pi_{\min} \cdot \min\left\{\frac{1}{1+\lambda}, \frac{\lambda}{1+\lambda}\right\} \right)^2 } \enspace.
\end{align*}
Finally, for the Hard-core model, the quantity $\pi_{\min}$ can be lower bounded as
\begin{align}\label{eq-lower-bound}
 \pi_{\min} &\geq \frac{\min\{1,\lambda^{|M|}\}}{ (1+\lambda)^{|M|} } \enspace,
\end{align}
which implies 
\begin{align*}
    C \leq 2|M|^2 \tp{1+\lambda+{1}/{\lambda}}^{2|M|+2} \enspace. 
\end{align*}
The above proves \eqref{eq:lemma-crude-boundA}.

We proceed with the proof of   \eqref{eq:lemma-crude-boundB}. Recall that, now,  we further assume that $\mu$ is $\eta$-spectrally 
independent. From  the definition of spectral independence, it is straightforward to see that $\pi$ is also spectrally 
independent. Hence, for any $0 \leq k \leq |M|-2$, any $S \subseteq M$ with $|S| = k$ and 
any feasible pinning  $\sigma \in \{\pm1\}^S$, the spectral radius of the influence matrix $\cI_\pi^{S,\sigma}$ satisfies 
\begin{align*}
\rho(\cI_\pi^{S,\sigma}) \leq \eta \enspace .
\end{align*}
Furthermore,  given any condition $\sigma$ on $S$, for any $v \notin S$, it holds that $\pi_v^{S,\sigma}(+1) \leq \frac{\lambda}{1+\lambda}$.
This implies that for any $u,v \notin S$, it holds that  
\begin{align*}
    \cI_\pi^{S,\sigma}(u,v) & \leq   \frac{\lambda}{1+\lambda} \enspace,
\end{align*}
which it turn implies that 
\begin{align*}
\frac{\rho(\cI_\pi^{S,\sigma})}{|M|-k-1}  & \leq \frac{\lambda}{1+\lambda} \enspace. 
\end{align*}
By \cite[Theorem 3.2]{feng2021rapid}, the spectral gap $\gamma$ can be lower bounded by
\begin{align*}
    \gamma & \geq \tp{\frac{1}{(1+\lambda)|M|}}^{2+2\eta} \enspace .
\end{align*}
For the sake of completeness Theorem 3.2 from \cite{feng2021rapid} can be found in the Appendix as 
\Cref{lemma-gap-SI}.

Plugging the above into ~\eqref{eq-C-bound}, we get that
\begin{align*}
    C & \leq  \frac{3 \log(1/\pi_{\min})}{|M|} \cdot ((1+\lambda)|M|)^{2+2\eta} \enspace.
\end{align*}
Finally, by the lower bound in~\eqref{eq-lower-bound}, we have
\begin{align*}
    C & \leq 3 \log\tp{1+\lambda+{1}/{\lambda}} \cdot ((1+\lambda)|M|)^{2+2\eta}\enspace. 
\end{align*}
The above concludes the proof of  \Cref{corollary-bound-hardcore}.
\end{proof}

\section{Proof of  \Cref{SI4HardCore} - Complete Spectral Independence}\label{sec:SI4HardCore}

In this section  we  establish that on typical instances of $\bG\sim G(n,d/n)$ the Hard-core model $\mu$
with fugacity $\lambda<\lambda_c(d)$ exhibits complete spectral independence in the way that is  specified in  
\Cref{SI4HardCore}.

As a first step we establish  spectral independence bounds for the Hard-core model on a 
{\em fixed} graph of a given $d$-branching value (see \Cref{def:DBranchingVal}), for some $d> 1$. 
Critically, the fugacity  $\lambda$ is upper bounded by  $\lambda_c(d)$.

We introduce   the non-homogenious Hard-core model, i.e., every vertex $u$ has its own fugacity 
$\lambda_u$.   Specifically, for $\blambda=(\blambda_v)_{v \in V}$ such that $\blambda_v\in \mathbb{R}_{>0}$, 
for all $v\in V$,   we define the non-homogenious Hard-core mode $\mu_{G,\blambda}$ such that for 
any $\sigma\in \{\pm\}^V$ we have that
\begin{align}
\mu_{G,\blambda}(\sigma) &\propto \textstyle \prod_{v\in V:\sigma_v=+} \blambda_v \enspace.
\end{align}
For the sake brevity, in what follows, we let $\lnorm \blambda \rnorm_{\infty}$ denote the maximum value over 
$\{\blambda_v\}_{v\in V}$.

Note that  we employ the {\em potential method} in order to establish  our  spectral independence. 
Specifically, we use the following lemma from   \cite{SSSY17}.

\begin{lemma}[\text{\cite{SSSY17}}]\label{lemma-SSSY}
Let $d > 1$ and $\lambda > 0$ be parameters satisfying $\lambda < \lambda_c(d)$. Let $\chi = \chi(d) \in (1,2)$ be a parameter 
defined as $\chi = (1 - \frac{d-1}{2}\log(1 + \frac{1}{d-1}))^{-1}$. Let $a=\frac{\chi}{\chi - 1}$ and $\Phi(x) = \frac{1}{\sqrt{x(x+1)}}$. 
There exists $0 < \kappa = \kappa(\lambda) < 1/d$ such that the following holds for any integer $k \geq 1$: for any 
$x_1,x_2,\ldots,x_k \geq 0$ and $x = \lambda \prod_{i=1}^k\frac{1}{1+x_i}$, it holds that 
$\Phi(x)^a \sum_{i=1}^k(\frac{x}{(1+x_i)\Phi(x_i)})^a \leq \kappa^{a/\chi}$.
\end{lemma}

We derive the following general result for  spectral independence for the non-homogenous Hard-core model
on a fixed  graph $G$ it terms of its $d$-branching value.

\begin{theorem}\label{lemma-SI}  
Let $d > 1$ be a real number. For the graph  $G=(V,E)$,  let $\mu_{G,\blambda}$  be the non-homogenious Hard-core model 
such that $\lnorm \blambda \rnorm_{\infty} <\lambda_c(d)$.

For any  $\alpha>0$ such that the $d$-branching value  $S_v \leq \alpha$ for all $v \in V$ the following is true: 
  $\mu_{G,\blambda}$  is $\eta$-spectrally independent for
\begin{align*}
    \eta     &\leq  W \cdot  \alpha^{1/\chi} \enspace, 
\end{align*}
 where   $W=W(d, \lnorm \blambda \rnorm_{\infty})$, while the quantity  $\chi=\chi(d)\in (1,2)$ is from Lemma \ref{lemma-SSSY}.
\end{theorem}

Note that  \Cref{lemma-SI} is about spectral independence, of the non-homogenious Hard-core model,  
while    \Cref{SI4HardCore} is about complete spectral independence of the (homogenious) Hard-core model.

We  prove  \Cref{SI4HardCore} from   \Cref{lemma-SI} and \Cref{lemma-d-factor}. 

\begin{proof}[Proof of \Cref{SI4HardCore}]
Recall that for \Cref{SI4HardCore} we assume that $\lambda<\lambda_c(d)$.  This implies that we
can choose  $\tilde{d}$ such that $\tilde{d}>d$, while 
\begin{align}\nonumber
\lambda<\lambda_c(\tilde{d}) \enspace.
\end{align}
Note that the  inequality is strict. The above follows by noting that the function $\lambda_c(x)$
 is continuous and  strictly decreasing for $x > 1$.
Also,  choose $\tilde{\lambda}$ such that
\begin{align}\nonumber
\lambda<\tilde{\lambda}<\lambda_c(\tilde{d}) \enspace.
\end{align}
Note that both  quantities $\tilde{\lambda}$ and $\tilde{d}$ depend only on $d,\lambda$, i.e., we have that
$\tilde{d}=\tilde{d}(d,\lambda)$ and $\tilde{\lambda} = \tilde{\lambda}(d,\lambda)$.

 We prove the theorem by showing that with probability $1-o(1)$ over the instances of $\bG$ we have the following:
 the Hard-core model on $\bG$ with fugacity $\lambda$, denoted as $\mu_{\bG}$,  is $(\eta,s)$-completely spectrally independent, where
\begin{align}
\eta&=B \cdot (\log n)^{\frac{1}{r}}  &\textrm{and} &&s&={\tilde{\lambda}}/{\lambda}-1 \enspace,
\end{align}
for constants   $B=B(d,\lambda)>0$ and  $r=r(d) \in (1,2)$.

Consider first the fixed  graph $G=(V,E)$ and external fields $\vec{\phi} \in (0,1+s]^V$. 
It is straightforward  that the distribution $\vec{\phi} * \mu_G$ corresponds to the  non-homogenious 
Hard-core model $\mu_{G, \blambda}$ such that $\lnorm \blambda\rnorm_{\infty} \leq \tilde{\lambda} <\lambda_{c}(\tilde{d})$. 

\Cref{lemma-SI} implies the following for $\mu_{G, \blambda}$:  for any  $\alpha>0$ such that the $\tilde{d}$-branching value  
$S_v \leq \alpha$ for all $v \in V$,  the distribution  $\mu_{G, \blambda}$ is $\zeta$-spectral independent where 
\begin{align}\label{eq:SIBoundF(AW)}
\zeta \leq   W \cdot  \alpha^{1/\chi} \enspace, 
\end{align}
for  $W=W(\tilde{d}, \lnorm \blambda \rnorm_{\infty})$ specified in \Cref{lemma-SI}, while the quantity  $\chi=\chi(\tilde{d}) \in (1,2)$ is from Lemma \ref{lemma-SSSY}.

The above imply that  $\mu_G$, the (homogenous) Hard-core model on $G$ with fugacity $\lambda$ 
is $(B \cdot  \alpha^{1/\chi}, s)$-completely spectrally independent.

Note that since   $W$ depends on $\tilde{d}$ and $\tilde{\lambda}=\lnorm \blambda \rnorm_{\infty}$, 
which in turn depend on  $d,\lambda$,  we have  that $W$ depends only on $d$ and $\lambda$. With the same argument,
the parameter $\chi$, above, depends only on $d$ and $\lambda$.

As far as $\bG \sim G(n,d/n)$ is concerned, we work as follows: 
since,  $\tilde{d}>d$, \Cref{lemma-d-factor} implies that with probability $1-o(1)$ over the instances
of $\bG$ the $\tilde{d}$-branching value  $S_v$, for all $v \in V$, satisfies $S_v\leq \log n$.
In light of this observation, the theorem follows from \eqref{eq:SIBoundF(AW)} by  having 
$B=W(\tilde{d}, \lnorm \blambda \rnorm_{\infty})$ and $r=\chi(\tilde{d})$.
\end{proof}

\subsection{A Bound on the Eigenvalue via Weighted Total Influence - Proof of \Cref{lemma-SI}}

We use the following result to prove \Cref{lemma-SI}.

\begin{theorem}\label{lemma-SI-refined}
Let $d > 1$ be a real number and $G=(V,E)$ be a graph. 
For  $\blambda \in \mathbb{R}^{V}_{>0}$, let $\mu_{G,\blambda}$ be the non-homogenious Hard-core model on $G$, 
while assume that $\lnorm \blambda \rnorm_{\infty} < \lambda_c(d)$. Also, let $\cI_G$ be the influence matrix induced by
$\mu_{G,\blambda}$.

For any $\alpha>0$ such that the $d$-branching value $S_v \leq \alpha$ for all $v \in V$ the following is true:
There exists a constant $W=W(d, \lnorm \blambda \rnorm_{\infty})$ such that  
\begin{align*}
    \forall r \in V, \quad \sum_{u \in V} \cI_G(r,u) \cdot \deg_G(u)^{1/\chi} \leq W \cdot \left( \alpha \cdot \deg_G(r)\right)^{1/\chi} \enspace,
\end{align*}
while $\chi = \chi(d) \in (1,2)$ is from  \Cref{lemma-SSSY}.
\end{theorem}
\Cref{lemma-SI-refined} is proved in \Cref{sec:lemma-SI-refined}.
Now, we are ready to prove \Cref{lemma-SI}.

\begin{proof}[Proof of \Cref{lemma-SI}]
To prove \Cref{lemma-SI} we focus on  the spectral radius of $\cI^{\Lambda, \sigma}_G$, i.e.,  $\rho(\cI^{\Lambda, \sigma}_G)$, 
and show that for any choice of $\Lambda\subset V$ and $\sigma \in \{\pm 1\}^\Lambda$ we have
\begin{align}\label{eq:Target4lemma-SI0}
\rho(\cI^{\Lambda, \sigma}_G) &\leq W \cdot  \alpha^{1/\chi} \enspace, 
\end{align}
\Cref{lemma-SI} follows immediately once we show the above.  

Before proving \eqref{eq:Target4lemma-SI0}, let us make  some useful observations. 
Suppose that we have  the non-homogenous Hard-core model with fugacities $\blambda\in \mathbb{R}^V_{>0}$ 
on the graph $G=(V,E)$, while at the  set of vertices $\Lambda$ we have the configuration $\tau$. 
Then, it is elementary  to verify that  this distribution is identical to the non-homogenous Hard-core model 
on the graph $G'=(V',E')$  with fugacities $(\blambda_v)_{v\in V'}$, such that $G'$   is obtained from 
$G$ by working  as follows: we remove from $G$ every vertex $w$ which
 either belongs to  $\Lambda$, or  has  a neighbour  $u\in \Lambda$ such  that $\tau(u)=+1$, i.e., $u$ is 
 ``occupied" under $\tau$. 

Additionally, consider the matrices   $\cI^{\Lambda,\sigma}_G$ and $\cI_{G'}$ induced by the aforementioned
Gibbs distributions, respectively. It is not hard to see that $\cI_{G'}$ is a principal submatrix of  $\cI^{\Lambda,\sigma}_G$
obtained by removing columns and rows that correspond to vertices $w\in V\setminus \Lambda$ that 
have a neighbour  $u\in \Lambda$ such  that $\tau(u)=+1$. Note that the rows and columns we remove
from $\cI^{\Lambda,\sigma}_G$ in order to obtain $\cI_{G'}$ consist of entries which are zero.  

Using the above observation, it is an easy exercise in  linear algebra to verify that
any additional eigenvalues that $\cI^{\Lambda,\tau}_G$ might have, compared  to $\cI_{G'}$, 
these can only be equal to zero.   
Hence, we derive the following relation for the spectral radii of the two matrices: 
\begin{align}\nonumber
\rho(\cI^{\Lambda,\tau}_G)= \rho(\cI_{G'}) \enspace.
\end{align}
Furthermore, note that the branching value is non-increasing when removing vertices.  Hence, 
 if we have $G, \blambda, d$ and $\alpha$ that satisfy the conditions in   \Cref{lemma-SI}, then  
 $G', (\blambda_v)_{v\in V'}, d$ and $\alpha$ satisfy the same conditions, as well. 

In light of all the above and without loss of generality, we can ignore the pinning and consider 
the influence matrix $\cI_G$. 

Hence, instead of proving \eqref{eq:Target4lemma-SI0} we consider the following equivalent problem.
Consider the graph $G=(V,E)$, while we have   $\alpha>0$ and $d>1$
such that the $d$-branching value  $S_v \leq \alpha$ for all $v \in V$.  Also, for 
$\blambda\in \mathbb{R}^{V}_{>0}$  such that $\lnorm \blambda \rnorm_{\infty}<\lambda_c(d)$
consider $\mu=\mu_G$ the non-homogenous Hard-core model on $G$ with fugacities $\blambda$, 
while let $\cI_{G}$ be the corresponding pairwise influence matrix (without boundary conditions).

It suffices to show that 
\begin{align}\label{eq:Target4lemma-SI}
\rho(\cI_G)\leq W \alpha^{1/\chi} \enspace,
\end{align}
for $W, \alpha, \chi$ specified in the statement of \Cref{lemma-SI}.

To this end, we use the matrix norm introduced in  \Cref{sec:SIWithBranching}. Specifically,
it is standard that 
\begin{align}\label{eq:RhoIGNB}
\rho(\cI_G) & \leq \textstyle \lnorm D^{-1}\cdot \cI_G\cdot D\rnorm_{\infty} \enspace, 
\end{align}
where 
$D$ is the diagonal matrix indexed by the vertices in $V\setminus \Lambda$ such that
\begin{align} \nonumber 
D(u,u) &=\begin{cases}
        \deg_G(v)^{1/\chi} &\text{if } \deg_G(v) \geq 1 \\
        1 &\text{if } \deg_G(v) = 0 \enspace,
\end{cases}
\end{align}
while $\chi = \chi(d) \in (1,2)$ is from  \Cref{lemma-SSSY}.

Noting that $\cI_G, D$ are  non-negative matrices, from the definition of the matrix norm $\lnorm \cdot \rnorm_{\infty}$, we have that
\begin{align}\nonumber  
 \lnorm D^{-1} \cdot \cI_G \cdot D \rnorm_{\infty}  &= 
 \max_{v\in V} 
 \sum_{u \in V} \frac{\cI_G(v,u)\cdot D(u,u)}{D(v,v)}  \nonumber \\
&=
 \max_{\substack{v\in V: \\ \deg_G(v)>0}} 
 \sum_{u \in V} \frac{\cI_G(v,u)\cdot \deg_G(u)^{1/\chi}}{\deg_G(v)^{1/\chi}} \nonumber \\ 
 & \leq W \alpha^{1/\chi} \nonumber
 \enspace.
\end{align}
The second equality follows from the observation that for the isolated vertices $v$, i.e., $\deg_G(v)=0$, we have
that $\cI_{G}(v,u)=0$ for all $u\in V$. The last  inequality is due to  \Cref{lemma-SI-refined}.
Hence, \eqref{eq:Target4lemma-SI} follows by plugging the above into \eqref{eq:RhoIGNB}.

 The theorem follows.
\end{proof}

\section{Bound the total influence via self-avoiding walk tree - Proof of \Cref{lemma-SI-refined}}\label{sec:lemma-SI-refined}
Let $G = (V,E)$ be a graph, while   assume there is a total order for the vertices in $V$. 

A self-avoiding walk (SAW) in $G$ is a path $v_1,v_2,\ldots,v_{\ell}$ in $G$ such that  $v_i \neq v_j$ for all $i \neq j$. 
Fix a vertex $r \in V$. We define  the SAW-tree $T_{\saw}(r)$,  the  {\em tree of self-avoiding  walks}, starting from $r$,  as follows:
Consider the set consisting of  every  walk  $v_0, \ldots, v_{\ell}$ in  the graph $G$ that emanates 
from vertex $r$, i.e., $v_0=r$,  while one of the following two holds
\begin{description}
\item[K.1]  $v_0, \ldots, v_{\ell}$ is a self-avoiding walk,
\item[K.2]  $v_0, \ldots, v_{\ell-1}$ is a self-avoiding walk, while there is $j\leq \ell-3$  such that $v_{\ell}=v_{j}$.
\end{description}
Each one of the walks in the set corresponds to a vertex in $T_{\saw}(r)$.  Two vertices in $T_{\saw}(r)$ are adjacent 
if the corresponding walks are adjacent.  Note that two  walks in the graph $G$ are considered to be  adjacent if one 
extends the other by one vertex \footnote{E.g.  the  walks    $P'=w_0, w_1, \ldots, w_{\ell}$ and 
$P=w_0, w_1, \ldots, w_{\ell}, w_{\ell+1}$ are adjacent with  each other.}.

We also use the following terminology:  for  vertex $u$ in $T_{\saw}(r)$ that corresponds to the walk 
$v_0, \ldots, v_{\ell}$ in $G$ we say that   ``$u$ is a {\em copy} of vertex $v_{\ell}$ in $T_{\saw}(r)$''.
For every vertex $v$ in the graph $G$, we use $C_{v}$ to denote the set of copies of $v$ in the SAW-tree $T_{\saw}(r)$.

For  $\blambda \in \mathbb{R}^{V}_{>0}$, consider the non-homogenious Hard-core model $\mu_{G,\blambda}$ on the graph $G$. 
We specify the non-homogenious Hard-core model $\mu_{T,\blambda}$ on $T_{\saw}(r)$ such that all the vertices in $C_v$ have 
fugacity $\blambda_v$, for all $v\in V$.

We use $\Lambda$ to denote the set of cycle-closing vertices in SAW-tree $T_{\saw}(r)$, i.e., those that correspond to the paths
of the kind {\bf K.2}.  Let $\sigma \in \{\pm1 \}^\Lambda$ denote the pinning  induced by the SAW-tree obtained by working as follows: 
for $z \in \Lambda$ that corresponds to the path $w_0, \ldots w_{\ell}$ we set $\sigma(z)$ such that 
\begin{enumerate}[(a)]	
\item $-1$ if $w_{\ell}>w_{\ell-1}$, 
\item $+1$ otherwise.
\end{enumerate}
 $\Lambda$ is a subset of the leaves of $T_{\saw}(r)$, and hence, $\sigma$ is a pinning of a subset of the leaves.
Note that, potentially, there are leaves in $T$ which do not belong to $\Lambda$. These are copies of
vertices in $G$ which are of degree 1.

The above construction gives rise to the  conditional Hard-core distribution $\mu_{T, \blambda}^{\Lambda,\sigma}$ on 
the SAW-tree $T = T_{\saw}(r)$.  Let $\cI_T^{\Lambda,\sigma}$ denote the influence matrix  induced by  
$\mu_{T, \blambda}^{\Lambda,\sigma}$.  

Recall that $\cI_G$ corresponds to the  the influence matrix induced by  $\mu_{G, \blambda}$. We have the following result that relates  the
influence matrices $\cI_G$ and  $\cI_T^{\Lambda,\sigma}$.
\begin{lemma}[\text{\cite[Lemma 8]{chen2020rapid} }]\label{lemma-inf-tree}
For every vertex $v \neq r$ in $G$, it holds that
\begin{align*}
    \cI_G(r,v) = \sum_{u \in C_v} \cI_T^{\Lambda,\sigma}(r,u) \enspace.
\end{align*}
\end{lemma}

The above result is very useful in that it allows us to study the influence matrix $\cI_G$ by means of the matrix 
$\cI_T^{\Lambda,\sigma}$ which is much simpler to analyse due to the tree underlying structure. 
In light of the above, we also use the following result from \cite{anari2020spectral}.
\begin{lemma}[\text{\cite[Lemma B.2]{anari2020spectral}}]\label{lemma:ProdInfluences}
Consider the tree $T=(V_T,E_T)$ and and let $\mu$ be a Gibbs distribution on $\{\pm 1\}^V$. For any 
three vertices $u,v,w\in V_T$ such that $u$ is on the path from $v$ to $w$, for any  
$M\subseteq V\setminus\{u,v,w\}$ and any $\tau\in \{\pm 1\}^{M}$ we have that
\begin{align}\nonumber 
\cI^{M,\tau}(v,w)&= \cI^{M,\tau}(v,u)\cdot \cI^{M,\tau}(u,w)\enspace.
\end{align}
\end{lemma}

Note that in \cite{anari2020spectral} the influence matrix is defined in a slightly different  way than what we have here.
Specifically,   the matrix  $\cI_G$ we define here can be obtained from the influence matrix in 
\cite{anari2020spectral} by taking the absolute value of its entries.  Even though 
\Cref{lemma:ProdInfluences}  was proved for the influence matrix in \cite{anari2020spectral}, it is straightforward
that it also holds for the influence matrix we define here. 

An observation that we use is that for all vertices $z\notin \Lambda$ in $T_{\saw}(r)$ 
have the following property: suppose that $z\in C_w$ for $w\in V$, then we have that
\begin{align}\label{eq:DegTVsDegG}
\deg_T(z)=\deg_G(w) \enspace .
\end{align}

In light of  Lemmas \ref{lemma-inf-tree} and \ref{lemma:ProdInfluences},  \Cref{lemma-SI-refined} follows by bounding  the total influence on the tree $T_{\saw}(r)$
from the root.   The following  is the main  technical result in this section.
\begin{proposition}\label{lemma-inf-on-tree}
Let $d >1$ be a real number and $T=(V_T,E_T)$ be a tree rooted at $r \in V$. 
Let $\blambda  \in \mathbb{R}^{V}_{>0}$ such that  $\lnorm \blambda \rnorm_{\infty}< \lambda_c(d)$.
Let $\mu_{T, \blambda}$ be the non-homogenius Hard-core model on $T$ with fugacity $\blambda$.
For any  $\alpha >0$  such that the $d$-branching value $S_r \leq \alpha$, the following is true:

There exists a constant $D=D(d,\lnorm \blambda \rnorm_{\infty})$ such that for any 
$\Lambda$ subset of the leaves of $T$, for any $\sigma\in\{\pm 1\}^{\Lambda}$ we have that
\begin{align*}
    \sum_{u \in V \setminus \{r\}} \cI_T^{\Lambda,\sigma}(r,u) \cdot  \deg_T(u)^{1/\chi} \leq D \cdot \left( \alpha\cdot \deg_T(r)\right)^{1/\chi} \enspace,
\end{align*}
where $\chi = \chi(d) \in (1,2)$ is defined in \Cref{lemma-SSSY}.
\end{proposition}

\begin{proof}[Proof of \Cref{lemma-SI-refined}]
Given graph $G$, fix a vertex $r \in V$.  Our focus is on bounding the weighted sum $\sum_{u \in V}|\cI_G(r,u)|\deg_G(u)^{1/\chi}$.

We  construct the SAW-tree $T = T_{\saw}(r)$ together with the pinning $\sigma$ on a subset of leaf vertices $\Lambda$, 
as we describe at the beginning of \Cref{sec:lemma-SI-refined}.
We have the non-homogenous Hard-core model on $T$ such that  every vertex $w$, a copy of $v \in V$ in $T$, 
has fugacity  $\blambda_v$.    Let $\cI_T^{\Lambda,\sigma}$ denote the  influence matrix that corresponds to this
distribution.

 \Cref{lemma-inf-tree} implies  that for any $u\in V$ we have
\begin{align}
\cI_G(r,u) \cdot \deg_G(u)^{1/\chi} &=  \deg_G(u)^{1/\chi} \cdot \sum_{v \in C_u} \cI_T^{\Lambda,\sigma}(r,v) \nonumber\\
&=  \deg_G(u)^{1/\chi} \cdot \sum_{v \in C_u\setminus \Lambda} \cI_T^{\Lambda,\sigma}(r,v) \label{BeforeA:eq-inf-graph-inf-tree} \\
  \tp{\text{by \eqref{eq:DegTVsDegG}}} &=   \sum_{v \in C_u\setminus \Lambda} \cI_T^{\Lambda,\sigma}(r,v) \deg_T(v)^{1/\chi} \nonumber \\ 
&=   \sum_{v \in C_u} \cI_T^{\Lambda,\sigma}(r,v) \deg_T(v)^{1/\chi}  \label{BeforeB:eq-inf-graph-inf-tree} \enspace .
\end{align}
In both \eqref{BeforeA:eq-inf-graph-inf-tree} and \eqref{BeforeB:eq-inf-graph-inf-tree} we use the observation that 
for all $w \in C_v \cap \Lambda$, we have $\cI_T^{\Lambda,\sigma}(r,w) = 0$, i.e., since 
the assignment of $w$ is fixed to  $\sigma(w)$,  the root has zero influence on $w$. Hence, we conclude that
\begin{align}\label{eq-inf-graph-inf-tree}
\sum_{u\in V}\cI_G(r,u) \cdot \deg_G(u)^{1/\chi}  &=\sum_{u\in V}   \sum_{v \in C_u} \cI_T^{\Lambda,\sigma}(r,v) \deg_T(v)^{1/\chi}\enspace. 
\end{align}

We let  $\widetilde{d}$ be the solution to the equation $\lambda_c(\widetilde{d}) = \frac{1}{2}(\lambda+\lambda_c(d))$. 
Note that $\widetilde{d}$ depends only on $d$ and $\lambda$. Also, we have that
\begin{align*}
(a)\quad \lambda &< \lambda_c(\widetilde{d}) <\lambda_c(d) &\textrm{and} & &(b)\quad \widetilde{d} &> d \enspace. 
\end{align*}
The inequality in (a) follows from the fact that $\lambda_c(\widetilde{d})$ is the average of $\lambda, \lambda_c(d)$ and 
$\lambda<\lambda_c(d)$. 
Also, (b) follows from  the observation that $\lambda_c(d)> \lambda_c(\widetilde{d})$ and that 
$\lambda_c(x)$ is monotonically decreasing  in $x>1$.

Let $\phi_r$ be the $\widetilde{d}$-branching value of $r$ in SAW-tree $T=T_{\saw}(G,r)$.
Let $\psi_r$ be the $d$-branching value of vertex $r$ in graph $G$. 
Note that there is a unique vertex in $G$ that corresponds to the root of the SAW-tree.


We claim that there exists a constant $K=K(d,\lambda)>1$ such that
\begin{align}\label{eq-branching}
  \phi_r  \leq K\cdot \psi_r \leq K\cdot \alpha,
\end{align}
where $\alpha$ satisfying $\alpha \geq \psi_r$ is specified in the statement of \Cref{lemma-SI-refined}.

Before showing that \eqref{eq-branching} is true, let us show how we can use it to prove  \Cref{lemma-SI-refined}.
Using  \Cref{lemma-inf-on-tree} with   $\widetilde{d}$-branching values (note that $\lambda < \lambda_c(\widetilde{d})$), 
we have the following:  there exists constants $D=D(\widetilde{d},\lambda)$ and  $\chi=\chi(\widetilde{d})$ such that 
\begin{align*}
   \sum_{u \in V} \cI_T^{(\Lambda,\sigma)}(r,u)\cdot \deg_T(u)^{1/\chi} &\leq D \left( \phi_r\cdot\deg_T(r)\right)^{1/\chi}\\
  \tp{\text{by \eqref{eq-branching}}}\quad &\leq  D \cdot   \left( K\cdot   \alpha\cdot \deg_T(r)\right)^{1/\chi}\enspace.
\end{align*}

In light of the above, \Cref{lemma-SI-refined} follows by plugging the above inequality into \eqref{eq-inf-graph-inf-tree}
and  setting $W=D \cdot K^{1/\chi}$.
Note that $W$ is a constant that depends only on $d$ and $\lambda$.

We conclude the proof of \Cref{lemma-SI-refined} by showing that~\eqref{eq-branching} is true. 
Let   $N^T_{r,\ell}$ be  the number of simple paths in $T$ of length $\ell$,  starting from the root $r$. 
Fix such a path $P =(v_0 = r,v_1,v_2,\ldots,v_{\ell})$ in the tree $T$. It follows from the definition of the SAW-tree $T$ that 
 $P$ corresponds to one of the following two types of  paths in  $G$.
\begin{description}
    \item[Type 1] $P$ is a simple path of length $\ell$ starting from $r$ in graph $G$;
    \item[Type 2] the prefix $v_0, v_1,\ldots,v_{\ell -1}$ is a simple path length $\ell-1$ starting from $r$ in graph $G$ and $v_{\ell}$ is a cycle-closing vertex such that $v_{\ell} = v_i$ for some $i \leq \ell - 3$.
\end{description}
Let $N^G_{r,\ell}$ be the number of length $\ell$ paths in $G$ that start from $r$ and are  of Type 1.  Let 
$\bar{N}^G_{r,\ell}$ be the number of length $\ell$ paths in $G$ that start from $r$ and are of Type 2.
We have the following relation
\begin{align}\label{eq-M-N}
   \forall \ell \geq 1,  \quad  N^T_{r,\ell} & =N^G_{r,\ell} + \bar{N}^G_{r,\ell} \ \leq\  N^G_{r,\ell} + \ell N^G_{r,\ell - 1} \enspace.
\end{align}
For the second inequality, we use the observation that $\bar{N}^G_{r,\ell}\leq \ell N^G_{r,\ell-1}$.

Recall that $\phi_r$ is the $\widetilde{d}$-branching value of $r$ in SAW-tree $T=T_{\saw}(G,r)$. We have that 
\begin{align*}
    \phi_r & = \sum_{\ell \geq 0}\frac{N^T_{r,\ell}}{\widetilde{d}^\ell} = 1 + \sum_{\ell \geq 1}\frac{N^T_{r,\ell}}{\widetilde{d}^\ell}\\
   \tp{\text{by }~\eqref{eq-M-N}}\quad &\leq 1 + \sum_{\ell \geq 1} \frac{N^G_{r,\ell}}{\widetilde{d}^\ell} + \sum_{\ell \geq 0}\frac{(\ell + 1)N^G_{r,\ell}}{\widetilde{d}^{\ell+1}}\\
   & =\sum_{\ell \geq 0}\frac{N^G_{r,\ell}}{\widetilde{d}^\ell}\tp{1 + \frac{\ell + 1}{\widetilde{d}}} \enspace.
\end{align*}
Since  $\widetilde{d} > d>1$ and $\widetilde{d}$ is determined by $d$ and $\lambda$,   there exists a constant $K=K(d,\lambda) \geq 1$
such that 
\begin{align*}
\forall \ell \geq 0, \quad   \frac{1}{\widetilde{d}^\ell}  \tp{1 + \frac{\ell + 1}{\widetilde{d}}} \leq \frac{K}{d^\ell} \enspace. 
\end{align*}
In turn, the above   implies that 
\begin{align*}
     \phi_r& \leq  K \cdot \sum_{\ell \geq 0} \frac{N^G_{r,\ell}}{d^\ell} = K\cdot \psi_r   \leq  K\cdot   \alpha.
\end{align*}
This proves~\eqref{eq-branching}.

All the above conclude the proof of \Cref{lemma-SI-refined}.
\end{proof}

\subsection{Proof of \Cref{lemma-inf-on-tree}}

For  $\Lambda \subseteq V_T\setminus\{r\}$ and   $\sigma \in  \{\pm 1\}^{\Lambda}$, recall from
\Cref{sec:StabilityMargRatio}  the  ratio of  marginals at the root $R^{K, \tau}_T(r)$  such that
\begin{align}\label{eq:DefOfR}
R^{\Lambda, \sigma}_T(r)=\frac{\mu_r(+1\ |\ \Lambda,  \sigma)}{\mu_r(-1\ |\ \Lambda,  \sigma)} \enspace.
\end{align}
Recall, also,  that  $\mu_r(\cdot \ |\ \Lambda,  \sigma)$ denotes the marginal of the Gibbs distribution 
$\mu_{T, \blambda}(\cdot \ |\ \Lambda,  \sigma)$ 
at the root $r$.  

For a vertex $u\in V_T$,  we let $T_u$ be the subtree of $T$ that includes $u$ and all its descendents. 
We always assume  that the root of $T_u$ is the vertex $u$.  
With a slight abuse of notation,  we let  $R^{\Lambda, \sigma}_u$ denote the ratio of marginals at the root for 
the subtree $T_u$, where the Gibbs distribution is, now,  with respect to $T_u$, while we impose the boundary 
condition $\sigma(K\cap T_u)$.

Letting $v_1,v_2,\ldots,v_k$ denote the children of $v$, it is standard to get the following recursion 
\begin{align}\label{eq-recursion}
    R^{\Lambda,\sigma}_v &= \blambda_v \prod_{i=1}^k \frac{1}{1+R^{\Lambda,\sigma}_{v_i}} \enspace.
\end{align}

Let $A\subseteq V_T$ includes all  $w\in V_T$ such that the path from the root $r$ to $w$
there is a vertex $v\in \Lambda$. Also, let $B$ includes all $w\in V_T$ such that 
$w$ has a child $u \in \Lambda$ and  $\sigma(u) = +1$. Define the set 
\begin{align*}
F &= \Lambda \cup A\cup B \enspace.
\end{align*}
Given the condition $(\Lambda,\sigma)$, it is standard to show that $\cI_{T}(r,w)=0$, for any 
 $w \in F$.   We call the vertices in $V_T \setminus F$ the \emph{free} vertices.

We define a set of parameters $(\alpha_v)_{v \in V}$ as follows: for each leaf vertex $v$ we have  $\alpha_v = 1$. For every non-leaf vertex 
$v$, let $v_1,v_2,\ldots,v_k$ denote the children of $v$ and define $\alpha_v = 1 + \frac{1}{d}\sum_{i=1}^k \alpha_{v_i}$. 
It is straightforward to verify that $\alpha_v$ is the $d$-branching value of $v \in V$ in the subtree $T_v$ rooted at $v$. 
Hence, it holds that $\alpha_r \leq  \alpha$.

For the sake of brevity, in  what follows, we abbreviate $\cI^{\Lambda, \sigma}_T$ and $R^{\Lambda,\sigma}_v$  to 
$\cI_T$ and $R_v$, respectively. 
Let $L(h)$ denote the set of all vertices at distance $h$ from the root $r$. Let $\chi \in (1,2)$ be the parameter and 
$\Phi(\cdot)$ be the function in \Cref{lemma-SSSY}. We claim that
\begin{align}\label{eq-levels}
 \forall h \geq 1, \quad   \sum_{v \in L(h) \setminus F}\tp{\frac{\alpha_v}{\alpha_r}}^{1/\chi}\frac{\cI_{T}(r,v)}{R_v \Phi(R_v)} 
 &\leq d^{1/\chi} \sqrt{\deg_T(r)}  (d\kappa)^{(h-1)/\chi} \enspace,
\end{align}
where
\begin{align}\label{eq-new-kappa}
\kappa = \sup_{0 <x \leq \lnorm \blambda\rnorm_{\infty}}\kappa(x) < 1/d,    
\end{align}
is the parameter, where the function $\kappa(x)$ is specified in \Cref{lemma-SSSY}. 
Hence, the parameter $\kappa$ in~\eqref{eq-levels} depends only on $\lnorm \blambda\rnorm_{\infty}$.
We remark that all the ratios in the above inequality are well-defined because $\alpha_r \geq 1$ and for any free 
vertex $v \notin F$, it holds that $0 < R_v < \infty$.

Before proving that \eqref{eq-levels} is true, we show how we can use it to prove the proposition. 
Note that $\alpha_v \geq 1$ for all $v \in V_T$. Since $d > 1$, we have
\begin{align*}
    \alpha_v \geq 1 + \frac{1}{d}(\deg_T(v)-1)  \geq \frac{\deg_T(v)}{d} \enspace .
\end{align*}
Note that $\cI_{T}(r,v)=0$ for all fixed $v \in F$. 
Hence, the weighted total influence can be bounded as follows
\begin{align}
 \sum_{v \in V_T: v\neq r}  \cI_{T}(r,v) \cdot \deg_T(v)^{1/\chi}&=  \sum_{v \in V_T \setminus F: v\neq r}  \cI_{T}(r,v) \cdot \deg_T(v)^{1/\chi} \nonumber \\
 &\leq \sum_{v \in V_T \setminus F: v\neq r} |\cI_{T}(r,v)| (d\alpha_v)^{1/\chi} \nonumber \\
 &= \tp{d\alpha_r}^{1/\chi} \sum_{v \in V_T \setminus F: v\neq r} \tp{\frac{\alpha_v}{\alpha_r}}^{1/\chi}{|\cI_{T}(r,v)|} \nonumber \\
 (\text{since } R_v\Phi(R_v) \leq 1)\quad &\leq \tp{d\alpha_r}^{1/\chi} \sum_{v \in V_T \setminus F: v\neq r} \tp{\frac{\alpha_v}{\alpha_r}}^{1/\chi}\frac{|\cI_{T}(r,v)|}{R_v \Phi(R_v)}
 \label{eq:RPhiRvInequality}\\
 &= \tp{d\alpha_r}^{1/\chi} \sum_{h=1}^\infty \sum_{v \in L(h) \setminus F} \tp{\frac{\alpha_v}{\alpha_r}}^{1/\chi}\frac{|\cI_{T}(r,v)|}{R_v \Phi(R_v)}
 \enspace. \nonumber 
\end{align}
Note that in \eqref{eq:RPhiRvInequality}, the inequality $R_v\Phi(R_v) \leq 1$ follows from the definition of $\Phi(\cdot)$ in \Cref{lemma-SSSY}.
Plugging \eqref{eq-levels} into the above inequality, we get that
\begin{align}
 \sum_{v \in V: v\neq r}  \cI_{T}(r,v) \cdot \deg_T(v)^{1/\chi}&\leq \tp{d^2\alpha_r}^{1/\chi}\sqrt{\deg_T(r)} \sum^{\infty}_{h=1}  (d\kappa)^{(h-1)/\chi} \nonumber\\
  (\text{since}\  d \kappa < 1)\quad &\leq {(d^2 \alpha_r)^{1/\chi} \sqrt{\deg_T(r)}} \cdot \frac{1}{1-(d\kappa)^{1/\chi}} \nonumber\\
 \tp{\text{by } \frac{1}{2} < \frac{1}{\chi} < 1}\quad  & \leq \left(d^2 \alpha_r  \deg_T(r)\right )^{1/\chi} \cdot \frac{1}{1-(d\kappa)^{1/\chi}} \nonumber\\
 \tp{\text{by } \alpha_r \leq \alpha }\quad &\leq  \frac{d^{2/\chi} \alpha^{1/\chi}}{1-(d\kappa)^{1/\chi}} \cdot \deg_T(r)^{1/\chi} \enspace. \nonumber
\end{align}
The proposition follows by setting  the parameter 
$D = D(d, \lnorm \blambda\rnorm_{\infty}) = \frac{d^{2/\chi}}{1-(d\kappa)^{1/\chi}} $.

We conclude the proof of \Cref{lemma-inf-on-tree} by showing that \eqref{eq-levels} is true. 
We prove~\eqref{eq-levels} by  induction on $h$.

The base case corresponds to  $h = 1$. Let $v_1,v_2,\ldots,v_k$ denote all free children of 
the root $r$,  i.e., $v_i\notin F$. Also, let $\Gamma$ denote the set of all children of $r$. 
Consider the influence from $r$ to a free  child $v_i \notin F$.  It is standard  to show that for 
the Hard-core   model $\mu_{T, \blambda}$ we have that
\begin{align}\label{eq-cI-R}
    \cI_T(r ,v_i) & =   \frac{R_{v_i}}{1+R_{v_i}} \enspace .
\end{align}
Hence, we get that 
\begin{align*}
    \frac{\cI_{T}(r,v_i)}{R_{v_i} \Phi(R_{v_i})} = \frac{\sqrt{R_{v_i}(R_{v_i}+1)}}{R_{v_i}+1} = \sqrt{\frac{R_{v_i}}{R_{v_i}+1}} \leq 1 \enspace.
\end{align*}
The above implies that
\begin{align*}
 \sum_{v \in L(1) \setminus F}\tp{\frac{\alpha_v}{\alpha_r}}^{1/\chi}\frac{ \cI_{T}(r,v)}{R_v \Phi(R_v)} &\leq \frac{\sum_{i=1}^k \alpha_{v_i}^{1/\chi} }{\alpha_r^{1/\chi}} =   \frac{\sum_{i=1}^k \alpha_{v_i}^{1/\chi} }{(1+\frac{1}{d}\sum_{v_j \in \Gamma} \alpha_{v_j})^{1/\chi}}\\
 \tp{\text{by }\{v_1,v_2,\ldots,v_k\} \subseteq \Gamma}\quad&\leq    \frac{\sum_{i=1}^k \alpha_{v_i}^{1/\chi} }{(1+\frac{1}{d}\sum_{j=1}^k \alpha_{v_j})^{1/\chi}}\\
 &\leq     d^{1/\chi}\cdot \frac{\sum_{i=1}^k \alpha_{v_i}^{1/\chi} }{(\sum_{j=1}^k \alpha_{v_j})^{1/\chi}} \enspace.
\end{align*}
Let $q\geq 1$ be such that $\frac{1}{q}+\frac{1}{\chi} = 1$. Note that  $q = \frac{\chi}{\chi - 1}$, while,  since $\chi \in (1,2)$, we have that 
$\frac{1}{q} \leq \frac{1}{2}$. 
By H\"{o}lder's inequality, we get that 
\begin{align*}
 \sum_{i=1}^k \alpha_{v_i}^{1/\chi} \leq  k^{1/q} \tp{\sum_{i=1}^k \alpha_v}^{1/\chi} \leq \sqrt{\deg_T(r)}  \tp{\sum_{i=1}^k \alpha_{v_i}}^{1/\chi} \enspace.
\end{align*}
This implies that
\begin{align*}
 \sum_{v \in L(1) \setminus F}\tp{\frac{\alpha_v}{\alpha_r}}^{1/\chi}\frac{ \cI_{T}(r,v) }{R_v \Phi(R_v)} \leq   d^{1/\chi}  \sqrt{\deg_T(r)} \cdot (d\kappa)^0  \enspace.
\end{align*}
The above proves the base of the induction. 

We now focus on proving the induction step. For  $\ell \geq 1$,  suppose~\eqref{eq-levels} is true for $h=\ell$. 
We prove that \eqref{eq-levels} is also true for $h=\ell + 1$.

For any free vertex $v \in L(\ell + 1) \setminus F$, $v$'s father $u \in L(\ell)$ must be free ($u \notin F$) due to the definition of $F$.  
For any $u \in L(\ell) \setminus F$, let $u_1,u_2,\ldots,u_{k(u)}$ denote the free children of the vertex $u$, where $0\leq k(u) \leq \deg_T(u)$. 

From \Cref{lemma:ProdInfluences} we have that $\cI_{T}(r,u_i) = \cI_{T}(r,u)\cdot \cI_{T}(u,u_i)$ for all $[i] \in k(u)$. 
Hence, we get that
\begin{align}
    \sum_{v \in L(\ell+1)\setminus F}\tp{\frac{\alpha_v}{\alpha_r}}^{1/\chi}\frac{\cI_{T}(r,v)}{R_v \Phi(R_v)} &= 
    \sum_{u \in L(\ell) \setminus F}\tp{\frac{\alpha_u}{\alpha_r}}^{1/\chi}\frac{ \cI_{T}(r,u) }{R_u \Phi(R_u)} 
    \sum_{i=1}^{k(u)} \tp{\frac{\alpha_{u_i}}{\alpha_u}}^{1/\chi} \frac{\cI_{T}(u,u_i)}{R_{u_i} \Phi(R_{u_i})}R_u\Phi(R_u) \label{eq:Base4inf-on-treeInd}\enspace.
\end{align}
From~\eqref{eq-cI-R}, we have
\begin{align}
 \sum_{i=1}^{k(u)} \tp{\frac{\alpha_{u_i}}{\alpha_u}}^{1/\chi} \frac{\cI_{T}(u,u_i)}{R_{u_i} \Phi(R_{u_i})}R_u\Phi(R_u) &
 \leq \sum_{i=1}^{k(u)} \tp{\frac{\alpha_{u_i}}{\alpha_u}}^{1/\chi} \frac{R_u\Phi(R_u) }{(R_{u_i}+1) \Phi(R_{u_i})} \nonumber \\
 \tp{\text{by H\"{o}lder's inequality}}\quad &\leq \tp{\frac{\sum_{i=1}^{k(u)} \alpha_{u_i}}{\alpha_u}}^{1/\chi}\tp{\Phi(R_u)^q
 \sum_{i=1}^{k(u)}\tp{\frac{R_u }{(R_{u_i}+1) \Phi(R_{u_i})}}^q }^{1/q} \enspace, \label{eq:inf-on-tree:StepInd}
\end{align}
where  $q = (1-1/\chi)^{-1}$. 
Using the  recursion in~\eqref{eq-recursion}, \Cref{lemma-SSSY} and the definition of $\kappa$ in~\eqref{eq-new-kappa}, we have
\begin{align} \nonumber 
\textstyle \tp{\Phi(R_u)^q  \sum_{i=1}^{k(u)}\tp{\frac{R_u }{(R_{u_i}+1) \Phi(R_{u_i})}}^q }^{1/q}  &\leq \kappa^{1/\chi}\enspace.
\end{align}
Plugging the above into \eqref{eq:inf-on-tree:StepInd} we get that 
\begin{align}\label{eq:inf-on-treeIndStepFinal}
\sum_{i=1}^{k(u)} \tp{\frac{\alpha_{u_i}}{\alpha_u}}^{1/\chi} \frac{ \cI_{T}(u,u_i)}{R_{u_i} \Phi(R_{u_i})}R_u\Phi(R_u)  & 
\leq  \tp{\frac{\sum_{i=1}^{k(u)} \alpha_{u_i}}{\alpha_u}}^{1/\chi} \kappa^{1/\chi} \leq  (d\kappa)^{1/\chi} \enspace,
\end{align}
where the last inequality holds since
\begin{align*}
\alpha_u = 1 + \frac{1}{d}\sum_{w: w \text{ is a child of } u} \alpha_w   \geq \frac{1}{d}\sum_{i=1}^{k(u)}\alpha_{u_i} \enspace.
\end{align*}
Plugging \eqref{eq:inf-on-treeIndStepFinal} into  \eqref{eq:Base4inf-on-treeInd} we get that
\begin{align}
   \sum_{v \in L(\ell+1)\setminus F}\tp{\frac{\alpha_v}{\alpha_r}}^{1/\chi}\frac{\cI_{T}(r,v)}{R_v \Phi(R_v)} & \leq (d\kappa)^{1/\chi}   \sum_{u \in L(\ell) \setminus F}\tp{\frac{\alpha_u}{\alpha_r}}^{1/\chi}\frac{ \cI_{T}(r,u) }{R_u \Phi(R_u)}  \nonumber \\
   &\leq  d^{1/\chi} \sqrt{\deg_T(r)}  (d\kappa)^{\ell/\chi} \enspace.  \nonumber 
\end{align} 
This proves the induction step and concludes the proof ~\eqref{eq-levels}. 

The proposition follows.
\hfill $\Box$

\section{Monomer-Dimer Model - Proof of \Cref{MainResultMD}}\label{sec:MonoDiProofs}

Let us first introduce the notion of line graph. Given a graph $G=(V,E)$, we have the  {\em line graph} $L$ of $G$ such that  each vertex 
in $L$ is an edge in $G$ and $e,f \in E$ are adjancent in $L$ if and only  if $e \cap f \neq \emptyset$.  Also note that, if $\Delta$ is the maximum 
degree in $G$, then the maximum degree in $L$ is at most $2\Delta$.

We use, here the  standard observation that  the Monomer-Dimer model with edge weight $\lambda$, on 
the graph $G=(V,E)$, corresponds to the Hard-core model on the line graph $L$ with fugacity $\lambda$. 

We often view the Monomer-Dimer model as a distribution over $\{\pm 1\}^E$, where for any $X \in \{\pm 1\}^E$, any $e \in E$, $X_e = +1$ represents that $e$ is in the matching and $X_e = -1$ represents that $e$ is not in the matching. 

\Cref{MainResultMD} is a corollary of the following more general result.

\begin{theorem}\label{thmmatchingproof}
 For any constants  $\lambda > 0$, there exist two constants $M_1 = M_1(\lambda), M_2 = M_2(\lambda)$ such that for any
 graph $G=(V,E)$ with $n$ vertices, maximum degree 
 $\Delta \geq 2$ and $m \geq 10^6(1+\lambda+1/\lambda)^3 \Delta^{3/2}$ edges the following is true:
 
Let $\mu_G$ be the Monomer-Dimer model on $G$ with edge weight $\lambda$. Then the Glauber dynamics
on $\mu_G$ exhibits  mixing time such that 
  \begin{align*}
       T_{\rm mix} \leq (M_1 \Delta)^{M_2 \sqrt{\Delta}} n \log n.
  \end{align*}
\end{theorem}

Specifically, \Cref{thmmatchingproof} implies \Cref{MainResultMD} because with probability 
$1-o(1)$ over the instances of $G(n,d/n)$, the maximum degree  is
$\Delta = \Theta(\frac{\log n}{\log \log n})$  (e.g. see \Cref{lemma-MaxDegGnp}), while  the number of 
edges is  $\Theta(n) = \omega(\Delta^{3/2})$. Note that the bound on the number of edges is a simple application of
Chernoff's bound. 
Using \Cref{thmmatchingproof}, we know that with probability 
$1-o(1)$ over the instances of $G(n,d/n)$,
\begin{align*}
    T_{\rm mix} \leq n^{1 + C\sqrt{\frac{\log\log n }{\log n}}}, 
\end{align*}
for some constant $C$ depending only on $\lambda$ and $d$.

From now on, our focus shifts to  proving  \Cref{thmmatchingproof}.  
As in the case of the Hard-core model, we consider the more general non-homogenous version of the Monomer-Dimer model.
That is, consider the graph $G=(V,E)$ and let $\blambda = (\lambda_e)_{e \in E} \in \mathbb{R}^E_{\geq 0}$ be an assignment 
of weight to each edge of the graph $G$.
Let $\mu_{\blambda}$ be the distribution over all matchings $\sigma$ such that 
$\mu_{\blambda}(\sigma)  \propto \prod_{e \in E: \sigma(e) = +1}\lambda_e$.

We have the following result,  which can be derived directly from  \cite{bayati2007simple,chen2020optimal}.

\begin{lemma}
\label{lemmamatchingsi}
 For any graph $G=(V,E)$ with the maximum degree $\Delta$, any edge weights $\blambda = (\lambda_e)_{e \in E} \in \mathbb{R}^E_{\geq 0}$,  
 the Gibbs distribution of  Monomer-Dimer model specified by $G$ and $\blambda$ is 
 $\left(2\sqrt{1 + \lnorm \blambda \rnorm_{\infty} \Delta}\right)$-spectrally independent.
\end{lemma}

Theorem 6.1 in \cite{chen2020optimal} is identical to \Cref{lemmamatchingsi} with the only difference that it is for 
{\em homogenous}  edge weights, i.e. $\lambda_e = \lambda$ for all $e \in E$.
One can extend this theorem and obtain \Cref{lemmamatchingsi} by combining the 
spectral independence analysis in \cite{chen2020optimal} and the correlation decay analysis in \cite{bayati2007simple}.


Furthermore, we have the following corollary about Complete Spectral Independence.

\begin{corollary}\label{corollary-mc}
For any constant $\lambda > 0$ and any graph $G=(V,E)$ with $n$ vertices and maximum degree $\Delta \geq 2$,
the following is true:

Let $\mu$ be the Monomer-Dimer model on $G$ with edge weight $\lambda$. Then,  there exists a constant $K=4\sqrt{1+\lambda}$ such that  
$\mu_G$ is  $(\eta,\xi)$-completely spectrally independent 
for 
  \begin{align*}
    \eta &\leq K\sqrt{\Delta} &\text{ and }&& \xi = 1.
  \end{align*}
\end{corollary}
\begin{proof}
Consider $\mu_{\blambda}$,  the non-homogenous Monomer-Dimer model on $G$ with edge weights $\blambda \in (0, \lambda(1+\xi)]^E$ .
It suffices to show that $\mu_{\blambda}$  is spectrally independent with parameter $\eta$. Using \Cref{lemmamatchingsi} we have that
$\mu_{\blambda}$ is  $\left(2\sqrt{1 + \lnorm \blambda \rnorm_{\infty} \Delta}\right)$-spectrally independent. Hence, we have that
\begin{align*}
2\sqrt{1 + \lnorm \blambda \rnorm_{\infty} \Delta} &=
  2\sqrt{1 + \lambda(1+\xi)\Delta} = 2\sqrt{1 + 2\lambda \Delta } \leq K\sqrt{\Delta} \enspace.
\end{align*}
\end{proof}

We also derive   marginal stability results.
\begin{theorem}[Stability  Monomer-dimer Model]\label{Bounds4MarginalRatiosMatching}
For any constant  $\lambda> 0$, for the graph $G=(V,E)$ with $n$ vertices and maximum degree $\Delta \geq 2$, 
let $\mu$ be of the Monomer-Dimer model on $G$ with  edge weight $\lambda$. Then, we have that $\mu$ is 
$(\lambda+2)^3\Delta^2$-marginally stable.
\end{theorem}
\begin{proof}
Let $\zeta = (\lambda+2)^3\Delta^2$.
For  $\Lambda \subseteq E$, a feasible configuration $\tau \in \{\pm\}^{\Lambda}$ and $e\in E\setminus \Lambda$, 
let the ratio of Gibbs marginals at $e$ 
\begin{align}
R^{\Lambda,\tau}_G(e) &= \frac{\mu^{\Lambda,\tau}_e(+1)}{\mu^{\Lambda,\tau}_e(-1)} \enspace.
\end{align}
Recall that for marginal stability, we need to have  that for any $S \subseteq \Lambda$,
\begin{align}\label{eq-Mar-mc}
    R^{\Lambda,\tau}_G(e) &\leq \zeta &  \textrm{and} &&
R^{\Lambda,\tau}_G(e) &\leq \zeta  \cdot R^{S,\tau_S}_G(e) \enspace. 
\end{align}
The first bound is easy because $R^{\Lambda,\tau}_G(e) \leq  \mu^{\Lambda,\tau}_e(+1) \leq \frac{\lambda}{1+\lambda} \leq \zeta$. We focus on the second one.

Suppose that $e = \{u,w\}$. 
Let $N_u$ (resp.  $N_w$) be  the set of edges incident to $u$ (resp. $w$) except for the edge $e$. 
We may assume that none of the edges in $N_u \cup N_v$ is set to be $+1$ by $\tau$, 
as otherwise $R^{\Lambda,\tau}_G(e)=0$ and \eqref{eq-Mar-mc} holds trivially.

 Next, we proceed to derive  a lower bound for $\mu^{S,\tau_S}_e(+1)$. Let the set
$F_u = N_u \setminus S$ and $F_w = N_w \setminus S$.  Also, let  $F_e=F_u\cup F_w$. We call
$F_e$ the set of free edges, since it  corresponds  the set of edges that are not fixed under $\tau_S$.

Letting $\bsigma$ be distributed as in $\mu^{S,\tau_S}_{E \setminus S}$,   we have
\begin{align}\label{eq:MDBaseLBMu1}
    \mu^{S,\tau_S}_e(+1) &\geq \tp{\frac{\lambda}{1+\lambda}} \cdot \Pr[\,\forall f \in F_e, \ \bsigma(f) = -1] \enspace. 
\end{align}
We use $\Omega$ to denote the support of $\mu^{S,\tau_S}_{E \setminus S}$. We partition $\Omega$ into two parts
\begin{align*}
    \Omega_- &= \{\sigma \in \Omega \mid \forall f \in F_e, \ \sigma(f) = -1\}\\
    \Omega_+ &= \{\sigma \in \Omega \mid \exists f \in F_e, \ \sigma(f) = +1\}\enspace .
\end{align*}
For the moment, assume that $\Omega_+\neq \emptyset$, i.e., the set is non-empty.

For a configuration $\sigma \in \Omega_+$, note there can be at most two  edges  $f, f'\in F_e$ such
$\sigma(f)=\sigma(f')=1$. Hence, the number of edges in $F_e$ that are set to $+1$, under $\sigma$, 
is at least 1 and  at most 2. Furthermore,  since $\sigma$ is a matching,    if there are two  edges $f, f'\in F_e$ such
$\sigma(f)=\sigma(f')=1$, these must be in different sets, e.g., $f\in F_u$ and $f'\in F_w$.

For  $\sigma \in \Omega_+$, let $\eta \in \Omega_-$ be a configuration that agrees with $\sigma$ on the 
assignment of the edges  outside  $F_{e}$. Note that, since $\eta \in \Omega_-$, we have  $\eta(f)=-1$ for all $f\in F_e$. 

Furthermore, noting that $\eta$ and $\sigma$ differ only on the configuration of at most two edges,  we have that
\begin{align}\label{eq:MDMargHVsMargT}
    \mu^{S,\tau_S}_{E \setminus S}(\eta) &\geq  \frac{\mu^{S,\tau_S}_{E \setminus S}(\sigma)}{\max\{\lambda^2, 1\}} 
    \geq  \frac{\mu^{S,\tau_S}_{E \setminus S}(\sigma)}{1+\lambda^2} \enspace .
\end{align}

Finally, we note that $\sigma$ can be uniquely specified by $\eta$ and the edges in $F_e$ at which
the two configurations disagree, i.e., recall that we  assumed that $\sigma,\tau$ disagree only at $F_e$. 
Hence, using \eqref{eq:MDMargHVsMargT}, we have that 
\begin{align*}
    \frac{\sum_{\eta \in \Omega_-}\mu^{S,\tau_S}_{E \setminus S}(\eta) }
    {\sum_{\tau \in \Omega_+}\mu^{S,\tau_S}_{E \setminus S}(\tau) } \geq \frac{1}{(1+\lambda^2)(|F_u|+1)\times(|F_w|+1)} 
    \geq \frac{1}{(1+\lambda^2)\Delta^2} \enspace.
\end{align*}
The above implies the following:  for $\Omega_+\neq \emptyset$, we have that
\begin{align}
\Pr[\,\forall f \in F_e, \ \bsigma(f) = -1] \geq  \tp{\frac{1}{1+(1+\lambda^2)\Delta^2}} \enspace,
\end{align}
where recall that $\bsigma$ be distributed as in $\mu^{S,\tau_S}_{E \setminus S}$.
Furthermore, for the    case where $\Omega_+=\emptyset$ it is immediate that $\Pr[\,\forall f \in F_e, \ \bsigma(f) = -1]=1$.

Hence,  we have that
\begin{align*}
     \mu^{S,\tau_S}_e(+1) &\geq \tp{\frac{\lambda}{1+\lambda}}  \cdot \Pr[\,\forall f \in F_e, \ \bsigma(f) = -1] \geq  \tp{\frac{\lambda}{1+\lambda}} \cdot \tp{\frac{1}{1+(1+\lambda^2)\Delta^2}}.
\end{align*}
Since $R^{\Lambda,\tau}_G(e)$ is increasing in the value of  $\mu^{S,\tau_S}_e(+1)$, we use the above to get  that 
\begin{align*}
   R^{S,\tau_S}_G(e) \geq  \frac{\tp{\frac{\lambda}{1+\lambda}} \cdot \tp{\frac{1}{1+(1+\lambda^2)\Delta^2}}}{1-\tp{\frac{\lambda}{1+\lambda}} \cdot \tp{\frac{1}{1+(1+\lambda^2)\Delta^2}}} \geq  \frac{\tp{\frac{\lambda}{1+\lambda}} \cdot \tp{\frac{1}{1+(1+\lambda^2)\Delta^2}}}{1-\tp{\frac{\lambda}{1+\lambda}} \cdot \tp{\frac{1}{1+(1+\lambda^2)\Delta^2}}} \cdot \frac{R^{\Lambda,\tau}_G(e)}{\lambda}.
\end{align*}
In the second inequality we use the fact that $R^{\Lambda,\tau}_G(e) \leq \lambda$.

The second inequality in~\eqref{eq-Mar-mc} can be proved from the above and  noting that 
\begin{align*}
 \frac{1-\tp{\frac{\lambda}{1+\lambda}} \cdot \tp{\frac{1}{1+(1+\lambda^2)\Delta^2}}}{\tp{\frac{\lambda}{1+\lambda}} \cdot \tp{\frac{1}{1+(1+\lambda^2)\Delta^2}}}    \cdot \lambda \leq \zeta = (\lambda+2)^3\Delta^2. &\qedhere
\end{align*}
This concludes the proof of \Cref{Bounds4MarginalRatiosMatching}.
\end{proof}

Finally, we have the following bound on the approximate tensorisation of entropy.

\begin{lemma}\label{corollary-bound-mc}
Consider the Gibbs distribution $\mu$ of the Monomer-dimer model specified by $G=(V,E)$ and edge weight $\lambda > 0$.
For any $k\geq 1$ and $H\subseteq E$ such that $|H|=k$, any feasible pinning $\tau \in \{\pm\}^{E \setminus H}$, the conditional distribution $\mu^{E \setminus H,\tau}_H$
 satisfies the approximate tensorization of entropy with constant  
 \begin{align*}
     \mathrm{AT}(k) \leq   2k^2 \tp{1+\lambda+{1}/{\lambda}}^{2k+2} \enspace.
 \end{align*}
\end{lemma}

Similarly to what we have for \Cref{Bounds4MarginalRatiosMatching},  \Cref{corollary-bound-mc} follows directly from 
\Cref{corollary-bound-hardcore} (i.e., the first bound) by utilising the connection between the Monomer-Dimer model on
$G$ and the Hard-core model on its line graph $L$. For this reason, we omit the proof of \Cref{corollary-bound-mc}.

We are now ready to prove \Cref{thmmatchingproof}.
\begin{proof}[Proof of  \Cref{thmmatchingproof} ]
We set the parameters 
\begin{align*}
   \eta =4\sqrt{1+\lambda} \cdot \sqrt{\Delta}, \quad\xi = 1, \quad \zeta = (\lambda+2)^3\Delta^2 \enspace.
\end{align*}

By \Cref{corollary-mc} and \Cref{Bounds4MarginalRatiosMatching}, the Gibbs distribution $\mu$ is $(\eta,\xi)$-completely spectrally independent and $\zeta$-marginally stable. Set the parameter 
\begin{align*}
    \alpha = \min \left\{ \frac{1}{2\eta}, \frac{\log(1+\xi)}{\log(1+\xi) + \log 2\zeta} \right\} \enspace.
\end{align*}
It is elementary calculations to verify that 
\begin{align*}
    \frac{1}{\alpha} \leq 100(1+\lambda)\sqrt{\Delta} \enspace.
\end{align*}

Let $m = |E|$ and set $\ell= \lceil \theta m \rceil$, where 
 \begin{align}\label{eq-def-theta-mc}
     \theta = \frac{1}{400 \mathrm{e} \Delta(1+\lambda+1/\lambda)^2}  \enspace. 
 \end{align}
 Since we assumed that $m \geq 10^6(1+\lambda+1/\lambda)^3 \Delta^{3/2}$, it holds that $1/\alpha \leq \ell < m$.
 By \Cref{thm-EI},
 the Gibbs distribution $\mu$ satisfies the $\ell$ block factorisation of entropy with parameter 
 \begin{align}\label{boundonC}
  C = \tp{\frac{em}{\ell}}^{1+1/\alpha} \leq  \tp{\frac{e}{\theta}}^{1+1/\alpha} \leq \frac{1}{2}(A \Delta)^{B \sqrt{\Delta}} \enspace,
 \end{align}
 for some constants $A = A(\lambda)$ and $B = B(\lambda)$.
 
Let $\Omega$ denote the support of $\mu$.  For  any $f:\Omega\to\mathbb{R}_{>0}$ 
we have that
\begin{align*}
\entropy_{\mu}(f) \leq \frac{C}{\binom{m}{\ell}} \sum_{S \in \binom{E}{\ell}} \mu \tp{\entropy_S(f)} \enspace,
\end{align*}
where $C$ is the parameter in~\eqref{boundonC}, while recall that  $m = |E|$.

We need to consider the subgraph induced by subset of edges $S$.
For any $S \subseteq E$, let  $C(S)$ denotes the set  of connected components in $G(V,S)$ which contains at least one edge.
With a slight  abuse of notation, we use   $U \in C(S)$ to  denote the  set of edges in the component $U$. 

By the conditional independence property of the Gibbs distribution and \Cref{lemma-ent-product}, we have
\begin{align*}
   \entropy_{\mu}(f) &\leq   \frac{C}{\binom{m}{\ell}}  \sum_{S \in \binom{E}{\ell}} \sum_{U \in C(S)} \mu \tp{\entropy_U(f)} \nonumber\\
 \tp{\text{by~\Cref{{corollary-bound-mc}} }}\quad  &\leq  \frac{C}{\binom{m}{\ell}}  \sum_{S \in \binom{E}{\ell}} \sum_{U \in C(S)}  \mathrm{AT}(|U|)\sum_{e \in U}\mu[\entropy_e(f)] \nonumber\\
 &\leq  C  \sum_{e \in E}\mu[\entropy_e(f)] \sum_{k \geq 1} \mathrm{AT}(k) \Pr[|C_e| = k]\\
 &\leq  C \sum_{e \in E}\mu[\entropy_e(f)] \sum_{k \geq 1} \tp{2k^2 \tp{1+\lambda+{1}/{\lambda}}^{2k+2}} \Pr[|C_e| = k] \enspace,
\end{align*}
where $C_e$ is the connected component in $G(V,S)$ containing $e$, while $S$ is sampled from $\binom{E}{\ell}$ uniformly at random.
At this point  we need to bound the probability term  $\Pr[|C_e| = k]$. 

Consider the line graph $L$ of $G$. Note that the maximum degree $\Delta_L$ of $L$ is at most $2\Delta$.
Furthermore, let $v_e$ denote the vertex in $L$ that corresponds to the edge $e$ in $G$. 

Suppose we sample $\ell$ vertices $\hat{S}$ uniformly at random from graph $L$. 
Let $C(v_e)$ the component in $L[\hat{S}]$ that includes vertex $v_e$. 
Then, it is straightforward  that the probability $\Pr[|C_e| = k]$ is equal to  the probability $\Pr[|C(v_e)| = k]$.
Recall that $\Pr[|C_e| = k]$ refers to choosing  uniformly random edges from $G$ and $\Pr[|C(v_e)| = k]$
refers to choosing uniformly at random vertices from $L$.

From  \cite[Lemma 4.3]{chen2020optimal}, we have
\begin{align*}
  \Pr[|C_e| = k] &=\Pr[|C(v_e)| = k]  \leq \frac{\ell}{m} (2\mathrm{e} \Delta_L \theta)^{k-1} \leq  \frac{\ell}{m} (4\mathrm{e} \Delta \theta)^{k-1}\\
  &  \leq \tp{\frac{1}{100(1+\lambda+1/\lambda)^2}}^{k-1} \enspace,
\end{align*}
the last  inequality follows from  that $\ell = \lceil \theta m \rceil$  for $\theta$ defined in~\eqref{eq-def-theta-mc}.
The above   implies that
\begin{align*}
    \entropy_{\mu}(f) &\leq C \sum_{e \in E}\mu[\entropy_e(f)] \sum_{k \geq 1} \tp{2k^2 \tp{1+\lambda+1/\lambda}^{2k+2}} \tp{\frac{1}{100(1+\lambda+1/\lambda)^2}}^{k-1}\\
    &\leq (A\Delta)^{B \sqrt{\Delta}} \sum_{e \in E}\mu[\entropy_e(f)],
\end{align*}
where the last inequality holds by~\eqref{boundonC}.
Note that $m \leq n \Delta $. We have
\begin{align*}
     T_{\rm mix} & \leq \left\lceil (A\Delta)^{B \sqrt{\Delta}}  m \tp{\log \log \frac{1}{\mu_{\min}} + \log(2)+2  } \right \rceil \leq 
     (M_1\Delta)^{M_2 \sqrt{\Delta}} n \log n,
\end{align*}
where $M_1 = M_1(\lambda)$ and $M_2 = M_2(\lambda)$ are two constants depending only on $\lambda$.
\end{proof}

\section*{Acknowledgement}
Charilaos Efthymiou is supported by EPSRC New Investigator Award (grant no. EP/V050842/1)  and 
Centre of Discrete Mathematics and Applications (DIMAP), The University of Warwick.

Weiming Feng is supported by funding from the European Research Council (ERC) under the European Union’s Horizon 2020 research and innovation programme (grant agreement No. 947778). 

Weiming Feng would like to thank Heng Guo for the helpful discussions.

\bibliographystyle{alpha}
\bibliography{refs}

\appendix
\section{Bounds on Spectral Gap from Spectral Independence}

\noindent 
Let $\mu = \mu_G$ be a Gibbs distribution on graph $G=(V,E)$ with support $\Omega \subseteq \{\pm 1\}^V$.
Let $P$ denote the transition matrix of the Glauber dynamics on $\mu_G$.
It is well-known that $P$ has non-negative real eigenvalues $1 = \lambda_1 \geq \lambda_2 \geq \ldots \lambda_{|\Omega|}\geq 0$.
The spectral gap of $P$ is defined by $1 - \lambda_2$. 
We have the following relation between spectral independence and spectral gap.
\begin{lemma}[\text{\cite[Theorem 3.2]{feng2021rapid}}]\label{lemma-gap-SI}
Let $\eta \geq 0$ and  $0\leq \varphi < 1$ be two parameters. 
Let $G=(V,E)$ be a graph with $n=|V|$ vertices.
Support the Gibbs distribution $\mu_G$ on  $G=(V,E)$  satisfies that for every $0\leq k\leq n-2$, $\Lambda\subseteq V$ of size $k$ and feasible configuration $\tau\in \{\pm 1\}^\Lambda$,
\begin{align*}
 \rho(\vert\cI^{\Lambda,\tau}_{G}\vert)\leq \eta \qquad\text{and}\qquad \rho(\vert\cI^{\Lambda,\tau}_{G}\vert)\leq \frac{\varphi}{n-k-1}.
\end{align*}
The spectral gap of Glauber dynamics on $\mu_G$ is at least 
\begin{align*}
    1 - \lambda_2 \geq \frac{(1-\varphi)^{2+2\eta}}{n^{1+2\eta}}.
\end{align*}
\end{lemma}

\section{Structural Properties of $G(n,p)$}
\begin{lemma}\label{lemma-MaxDegGnp}
For  real numbers $d>0$ and  $\epsilon>0$, let $\Delta$ be the maximum degree of  the graph $\bG\sim G(n,d/n)$. 
Then the following is true:  With probability $1-o(1)$ the graph $\bG$ satisfies that
\begin{align}
(1-\epsilon)\frac{\log n}{\log\log n} \leq \Delta \leq (1+\epsilon)\frac{\log n}{\log\log n}. \nonumber
\end{align}
\end{lemma}
The above result is standard to derive,  using the first and second moment method.

\end{document}